\documentclass[12pt]{article}
\usepackage{amsmath,amsthm}
\usepackage{eufrak}
\usepackage[pdftex]{graphicx}

\newcommand{\ar}{\renewcommand{\arraystretch}{1}} 
\DeclareMathAlphabet{\bb}{U}{msb}{m}{n} \gdef\C{\bb C} \gdef\dZ{\bb
Z}   \gdef\dS{\bb S} \gdef\R{\bb R}
\gdef\K{\bb K} \gdef\BH{\bb H} \gdef\F{\bb F} \gdef\dO{\bb O}

\DeclareMathOperator{\End}{End} \DeclareMathOperator{\spin}{{\bf
Spin}} 
\DeclareMathOperator{\fD}{\mathfrak{D}}
\DeclareMathOperator{\Id}{Id} 
  
\DeclareMathOperator{\Ker}{Ker} 
  
\DeclareMathOperator{\Sym}{Sym} 
 
\DeclareMathOperator{\Ext}{Ext} \DeclareMathOperator{\Mat}{Mat}
 
\DeclareMathOperator{\Tr}{Tr} \DeclareMathOperator{\SL}{SL}
\DeclareMathOperator{\SO}{SO}\DeclareMathOperator{\SU}{SU}
\DeclareMathOperator{\Sp}{Sp} 
\DeclareMathOperator{\Hilb}{Hilb}

\newcommand{\bcirc}{\raisebox{0.5mm}{$\scriptstyle\bigcirc$}}

\newcommand{\cA}{\mathcal{A}}

\newcommand{\bcE}{\boldsymbol{\mathcal{E}}}
\newcommand{\bcK}{\boldsymbol{\mathcal{K}}}
\newcommand{\cP}{{\cal P}}
\newcommand{\cO}{{\cal O}}

\newcommand{\fB}{\mathfrak{B}}

\newcommand{\sA}{{\sf A}}
\newcommand{\sB}{{\sf B}}

\newcommand{\sH}{{\sf H}}

\newcommand{\sX}{{\sf X}}
\newcommand{\sY}{{\sf Y}}

\newcommand{\bsH}{{\boldsymbol{\sf H}}}

\newcommand{\bsT}{{\boldsymbol{\sf T}}}
\newcommand{\bsU}{{\boldsymbol{\sf U}}}
\newcommand{\bsZ}{{\boldsymbol{\sf Z}}}

\newcommand{\bx}{{\bf x}}

\newcommand{\bZ}{{\bf Z}}

\newcommand{\fA}{\mathfrak{A}}

\newcommand{\fa}{\mathfrak{a}}
\newcommand{\fb}{\mathfrak{b}}

\newcommand{\fh}{\mathfrak{h}}

\newcommand{\Lip}{\boldsymbol{\Gamma}}

\newcommand{\cl}{C\kern -0.2em \ell}

\newcommand{\e}{\mbox{\bf e}}

\newtheorem{thm}{Theorem}
\newtheorem{defn}{Definition}
\begin{document}
\title{Spinors in $\K$-Hilbert Spaces}
\author{V.~V. Varlamov\thanks{Siberian State Industrial University,
Kirova 42, Novokuznetsk 654007, Russia, e-mail:
varlamov@sibsiu.ru}}
\date{}
\maketitle
\begin{abstract}
We consider a structure of the $\K$-Hilbert space, where $\K\simeq\R$ is a field of real numbers, $\K\simeq\C$ is a field of complex numbers, $\K\simeq\BH$ is a quaternion algebra, within the framework of division rings of Clifford algebras. The $\K$-Hilbert space is generated by the Gelfand-Naimark-Segal construction, while the generating $C^\ast$-algebra consists of the energy operator $H$ and the generators of the group $\SU(2,2)$ attached to $H$. The cyclic vectors of the $\K$-Hilbert space corresponding to the tensor products of quaternionic algebras define the pure separable states of the operator algebra. Depending on the division ring $\K$, all states of the operator algebra are divided into three classes: 1) charged states with $\K\simeq\C$; 2) neutral states with $\K\simeq\BH$; 3) truly neutral states with $\K\simeq\R$. For pure separable states that define the fermionic and bosonic states  of the energy spectrum, the fusion, doubling (complexification) and annihilation operations are determined.
\end{abstract}

\section{Introduction}
\begin{flushright}
\begin{minipage}{18pc}{\small
No one fully understand spinors. Their algebra is formally understood but their geometrical significance is mysterious.

Michael Atiayh \cite[P.~430]{FG}
}\end{minipage}
\end{flushright}
The roots of Clifford algebra theory go back to Hamilton \textit{quaternions} \cite{Ham} and Grassmann {\it Ausdehnungslehre} \cite{Grass}. The algebra introduced by Clifford \cite{Cliff} is a generalization of the quaternion algebra to the case of multidimensional spaces. As a consequence of this generalization, a quaternion structure arises in the Clifford algebra, which is a tensor product of quaternion algebras, that is, a tensor product of four-dimensional algebras (the dimension of the quaternion algebra is 4). Studying the rotations of the $n$-dimensional Euclidean space $\R^n$, Lipschitz \cite{Lips} found that the group of rotations of the space $\R^n$ with the determinant $+1$ is represented by a \textit{spinor group}\footnote{For an $n$-dimensional pseudo-Euclidean space $\R^{p,q}$ ($n=p+q$), the Lipschitz group $\Lip_{p,q}=\left\{s\in\cl_{p,q}\;|\;\forall \bx\in\R^{p,q},\;
s\bx s^{-1}\in\R^{p,q}\right\}$ contains a spinor group $\spin(p,q)=\left\{s\in\Lip^+_{p,q}\;|\;N(s)=\pm 1\right\}$, where $\Lip^+_{p,q}=\Lip_{p,q}\cap\cl^+_{p,q}$, $\cl^+_{p,q}$ is an even subalgebra of the Clifford algebra $\cl_{p,q}$ of the space $\R^{p,q}$. The group $\spin(p,q)$ contains the subgroup $\spin_+(p,q)=\left\{s\in\spin(p,q)\;|\;N(s)=1\right\}$.}. As is known \cite{Roz55}, the groups of motions of $n$-dimensional non-Euclidean spaces $S^n$ are isomorphic to the groups of rotations of spaces $\R^{n+1}$. Since Clifford algebras are isomorphic to matrix algebras, \textit{spinor representations} of motions of spaces $S^n$ can be considered as representations of vectors in corresponding spaces. The vectors of these spaces are called \textit{spinors} of spaces $S^n$. The concept of spinor was introduced by Cartan \cite{Car13}. Van der Waerden notes \cite{Waer} that the name ``spinor'' was given by Ehrenfest with the appearance of the famous article of Uhlenbeck and Goudsmit \cite{UG} about a spinning electron. More precisely, the geometric meaning of the spinor representations of the motions of non-Euclidean spaces $S^n$ is that the coordinates of the spinors can be considered as the coordinates of the plane generators of the maximal dimension of the absolutes\footnote{An \textit{absolute} is a set of infinitely distant points of a non-Euclidean space.} of these spaces, and the spinor representations of the motion of these spaces coincide with those transformations of the spinors that correspond to the transformations of the absolutes during the movements. Thus, in a case important for physics, it is known that the connected group of motions of the three-dimensional non-Euclidean space $S^{1,2}$ (Lobachevsky space) is isomorphic to the connected group of rotations of the four-dimensional pseudo-Euclidean space $\R^{1,3}$ (Minkowski space-time), coinciding with the group of Lorentz transformations of special relativity. Therefore, the spinor representations of the connected group of motions of the space $S^{1,2}$ are at the same time spinor representations of the Lorentz group. It follows that each spinor of the space $S^{1,2}$ corresponds to some point of the absolute\footnote{The absolute of the Lobachevsky space $S^{1,2}$ is homeomorphic to the extended complex plane $\C\cup\{\infty\}$.} of the space $S^{1,2}$, and each point of the absolute of the space $S^{1,2}$ corresponds to an isotropic line of the space $\R^{1,3}$ passing through some point of this space. The described geometric interpretation of spinors and spinor representations was proposed by Cartan \cite{Car38} (see also \cite{Roz55}).

In his book \cite{Che54}, Chevalley notes that the construction of the concept of spinor given by Cartan was rather complicated, and a much simpler presentation of the theory, based on the use of Clifford algebras, was given by Brauer and Weyl in \cite{BW35}. In his presentation of the algebraic theory of spinors, Chevalley follows the article \cite{BW35}. Thus, along with the geometrical interpretation, an algebraic approach to the description of spinors and spinor representations appeared, which was further developed in \cite{Port,Cru,Lou}. The productivity and development of the algebraic approach shows that the spinor is primarily an object of algebraic nature. According to the algebraic definition, a spinor is an element of the minimal left ideal of the Clifford algebra\footnote{The first one to consider spinors as elements in a minimal left ideal of a Clifford algebra was M. Riesz \cite{Rie47}.} $\cl(V,Q)$, where $V$ is a vector space equipped with a non-degenerate quadratic form $Q$. For $n$ even, the minimal left ideal of $\cl(V,Q)$ corresponds to the \textit{maximal totally isotropic subspace}\footnote{A subspace $U$ of a space $V$ is called totally isotropic if the bilinear form $B(\alpha_i,\alpha_j)=0$ for all $\alpha_i,\alpha_j\in U$. A subspace $U\subset V$ of maximal dimension with the above property is called a maximal totally isotropic subspace.} $U\subset V$ of dimension $n/2$, i.e. it is isomorphic to a spinspace $\dS$ of dimension $2^{n/2}$ \cite{Abl01}.

Returning to the spinor representations of the Lorentz group, we see that the \textit{fundamental representation} of this group over the field of complex numbers $\F=\C$ acts in a two-dimensional spinspace, the vector of which is a two-dimensional spinor. And further, any finite-dimensional irreducible spinor representation of the Lorentz group can be factorized as a tensor product of two-dimensional fundamental representations. In turn, the spinspace of the fundamental representation is the minimal left ideal of the quaternion algebra (the biquaternion algebra $\C_2\simeq\C\otimes\BH$ in the case of the field $\F=\C$ and the quaternion algebras $\cl_{0,2}\simeq\BH$, $\cl_{1,1}\simeq\R(2)$, $\cl_{0,2}\simeq\R(2)$ in the case of the field $\F=\R$, while with the division ring $\K\simeq\BH$ for $\cl_{0,2}$ and $\K\simeq\R$ for $\cl_{1,1}$ and $\cl_{2,0}$). The elements of minimal left ideals of four-dimensional quaternionic algebras are two-component spinors. Here we have considered only one of the ways in which $\K$-structures (division algebras) arise in theoretical physics.

In a broader context, the algebraic formulation of quantum theory was first proposed by Jordan, von Neumann, and Wigner \cite{JNW34,Neu36}, as well as by Segal \cite{Seg47} (in terms of $C^\ast$-algebras). Subsequently, algebraic methods penetrated into quantum field theory and statistical mechanics \cite{Emh,BLOT,Hor86}. The authors of \cite{JNW34,Neu36} considered $\R$-, $\C$- and $\BH$-realizations of quantum theory on equal grounds. Historically, this was the first appearance of $\K$-structures in physics. In 1962, Dyson \cite{Dys62} proposed to consider all three $\K$-realizations of quantum theory ($\K=\R,\C,\BH$) within a single structure, which he called the ``threefold way''. More precisely, the Dyson threefold way describes how certain complex representations of groups can be seen as arising from real or quaternionic representations. And further, in 2012, Baez \cite{Baez} develops the threefold way within the framework of the theory of categories and $\K$-Hilbert spaces. The theory of $\K$-Hilbert spaces was previously proposed in \cite{Sol}. Another application of $\K$-structures in physics is related to the algebraic description of the standard model (Dixon algebra $\R\otimes\C\otimes\BH\otimes\dO$ \cite{Dix}, where $\dO$ is the octonion algebra\footnote{G\"{u}naydin and G\"{u}rsey \cite{GG74} considered the algebraic representation of the quark model in the framework of the octonion automorphism group $G_2$ ($G_2$ is a 14-dimensional exceptional Lie group containing $\SU(3)$ as a subgroup).}). The representation of discrete symmetries in the form of automorphisms of Clifford algebras was proposed in \cite{Var01,Var04,Var05,Var15} ($PT$- and $CPT$-groups). In this case, the $\K$-structure of the pseudo-automorphism of the Clifford algebra, which defines the charge conjugation $C$, plays an essential role: the division ring $\K\simeq\C$ corresponds to charged states (particles), $\K\simeq\BH$ to neutral states, and $\K\simeq\R$ corresponds to truly neutral (Majorana) states. The definition of fermionic and bosonic representations of the Lorentz group in terms of tensor products of two-dimensional spinspaces (minimal left ideals of quaternion algebras) and the action of the corresponding $CPT$ groups for higher spin fields was given in \cite{Var11,Var14}.

A turning point in the development of algebraic quantum theory was Haag's 1957 paper \cite{Haag}, in which the concept of the algebra of local observables $\fA(\cO)$ was first introduced. Haag's axiomatics is based on the principle of locality and the concept of a local observable, that is, an arbitrary physical quantity that can be measured experimentally in a limited region $\cO$ of space-time\footnote{There are two variants of the local algebraic approach: the concrete one (or \textit{Haag-Araki theory}), in which the local algebras are von Neumann algebras $\overline{\fA}(\cO)$ in some Hilbert space $\sH$, and the abstract one (or \textit{Haag-Kastler theory}), in which the local algebras are abstract $C^\ast$-algebras $\fA(\cO)$. The difference between these two variants is only meaningful from a constructive point of view. Namely, it may happen that the algebra of observables is constructed before its physical representation $\pi$ is chosen, then the abstract-algebraic point of view (Haag-Kastler theory) is preferable. When the physical representation $\pi$ is fixed, then the abstract $C^\ast$-algebras $\fA(\cO)$ can be considered as ``concrete'' von Neumann algebras $\overline{\pi(\fA(\cO))}$ (local observables defined in the weak operator topology of the physical representation).}. The main task of the local algebraic approach is a physically acceptable description of ``relativistic quantum objects'', in connection with which repeated attempts have been made to postulate the axiom of the corpuscular interpretation. In local quantum theory, the search for the criterion of corpuscular interpretation proceeds from the general interpretation of a particle as an ``asymptotically stable localization center''. However, as noted by Haag and Buchholz \cite{BH}, an adequate way to mathematically express such a concept of a particle has not yet been found. In this connection, the question naturally arises about the adequacy of the use of the term ``particle'' in relation to the quantum micro-object. With this question (``what is an elementary particle?'') Heisenberg \cite{Heisen}, Schr\"{o}dinger, Markov were asked. This term, more precisely, a visual spatial image, usually associated with the concept of a particle, is unconsciously translated to objects in microcosm, which, in general, do not obey the classical space-time description. As a result, there is a wide range of fuzzy and blurred ideas about the quantum micro-object. In a recent paper \cite{Wol20}, Wolchover summarizes these views. Let's list some of them: 1) a particle as a result of the collapse of the wave function; 2) a particle as a perturbation of the quantized field; 3) a particle as an irreducible representation of the group\footnote{This interpretation goes back to Wigner work \cite{Wig39}, in which an elementary particle is described by an irreducible \textit{representation}
of the Poincar\'{e} group $\cP$. On the other hand, in accordance with $\SU(3)$-theory, an elementary particle is described by the \textit{vector} of the irreducible representation of the group $\SU(3)$. For example, in the so-called Gell-Mann ``eightfold way'' \cite{GN64}, hadrons are represented by vectors of the eight-dimensional regular representation $\Sym^0_{(1,1)}$ of the group $\SU(3)$. In order to build a bridge between these interpretations (between \emph{representations} of the group $\cP$ and \emph{vectors} of the representations $\Sym^0_{(1,1)}$, $\Sym^0_{(1,4)}$, $\ldots$), we introduce a $\K$-Hilbert state space in the section 5, where each cyclic vector of this space defines an irreducible representation of the group $\SL(2,\C)$ (see also \cite{Var14,Var15,Var16}). In this approach, the concept of symmetry takes on a dominant role.}; 4) a particle as a string vibration; 5) a particle as a deformation of the information ocean\footnote{Associations of the oceanic plane often arise when it comes to the nature of the fundamental substance that fills the Universe, here it is enough to recall the \textit{Dirac Sea} or the \textit{neutrino sea} of Pontecorvo-Smorodinsky \cite{PS61}.} (Wheeler's ``it from bit''). The latter interpretation goes back to the Ur-hypothesis of von Weizs\"{a}cker \cite{Wa55,Wa92}, as well as to the hypothesis of the space-time code of Finkelstein \cite{Fin}. 
The mathematical structure of the information bit is identical to the two-component spinor. In turn, the two-component spinor describes the neutrino. 

In this paper, we consider an algebraic formulation of a quantum theory with a binary structure. Following Heisenberg \cite{Heisen51}, the main observable at the fundamental level (microlevel) is the energy that corresponds to the Hermitian operator $H$. The $C^\ast$-algebra $\fA$ of observables consists of the energy operator $H$ and the generators of the group $\SU(2,2)$ attached to $H$ (the twofold covering of the conformal group $\SO_0(2,4)$), which form a general system of eigenfunctions with $H$. The spectrum of states (the spectrum of matter) is generated by the Gelfand-Naimark-Segal construction. The pure separable states of the matter spectrum are given by the cyclic vectors of the $\K$-Hilbert space, where $\K=\R,\C,\BH$. According to the mass spectrum, the pure separable states $\omega$ form a physical $\K$-Hilbert space $\bsH_{\rm phys}(\K)$, where each cyclic vector $\left|\psi\right\rangle\in\bsH_{\rm phys}(\K)$ corresponds to a certain observed state (``particle''). In section 6, we introduce the coherent subspaces $\bsH^{(b,\ell)}_{\rm phys}(\K)$ of the physical $\K$-Hilbert space $\bsH_{\rm phys}(\K)$, where $b$ and $\ell$ are the baryon and lepton numbers. \textit{Algebraic quantization}, implemented by the GNS construction within the framework of the concept of cyclic vectors of the $\K$-Hilbert space and the \textit{energy interpretation} of the $C^\ast$-algebra, naturally leads to an understanding of the state (particle) as a \textit{\textbf{quantum of energy}}. In this case, the \textit{minimal energy quantum} corresponds to the fundamental state of the $C^\ast$-algebra (the minimal left ideal of the quaternion algebra), and the tensor products of the fundamental states form \textit{fermionic states} for an odd number of cofactors and \textit{bosonic states} for an even number. In section 7, it is shown that all states (fermionic and bosonic) from $\K$-subspaces ($\K=\R,\C,\BH$) are derived structures that are obtained from the fundamental states through fusion and doubling operations.

\section{$C^\ast$-algebras of Observables and Gelfand-Naimark-Segal Construction}
In this section, we will briefly consider the main definitions concerning the theory of $C^\ast$-algebras (see also \cite{Emh,BLOT,Hor86,BF16}). It is known that \textit{any quantum system} is characterized by a set of observational data that can be obtained as a result of the corresponding measurement process. The physical quantities obtained as a result of the measurement are \textit{observables} of the quantum system. The set of observables forms an algebra $\fA$, in which the operation of multiplying the observables is defined and their linear superpositions are given. In general, an algebra $\fA$ is an associative noncommutative $C^\ast$-algebra with unity over a field of complex numbers $\F=\C$. Further, the algebra $\fA$ is endowed with a conjugation operation, that is, there is an anti-linear involution $\ast:\fA\rightarrow\fA$ such that $(\fa^\ast)^\ast=\fa$ for any element $\fa\in\fA$. The norm $\|\cdot\|$ on $\fA$ is defined as follows: for any $\fa,\fb\in\fA$ the inequality $\|\fa\fb\|\leq\|\fa\|\|\fb\|$ is true, as well as $\|\fa^\ast\fa\|=\|\fa\|^2$, that is, $\|\fa^\ast\|=\|\fa\|$.

In the case of an $n$-level quantum system, the algebra $\fA$ can be identified with the $C^\ast$-algebra $\Mat_n(\C)$ of complex $n\times n$ matrices. In this case, $\ast$-operation coincides with the Hermitean conjugation $M^\ast=M^\dag$ for any element $M\in\Mat_n(\C)$, and the norm $\|M\|$ is given by the largest eigenvalue of the product $M^\dag M$. In the case of an infinite number of degrees of freedom, the $C^\ast$-algebra is the algebra $\fB(\sH)$ of all bounded operators on an infinite-dimensional Hilbert space $\sH_\infty$ ($\sH_\infty$ is a Banach space with a countable base that is everywhere dense in $\sH_\infty$).

The explicit relationship between the algebra $\fA$ and measurement data is given by the concept of state $\omega$, by which the expected value $\omega(\fa)$ of the observable $\fa\in\fA$ can be determined. A state $\omega$ on the $C^\ast$-algebra $\fA$ is a linear map $\omega:\fA\rightarrow\C$ that is positive, $\omega(\fa^\ast\fa)\geq 0$, $\forall\fa\in\fA$, and normed, $\omega(1_\fA)=1$, where $1_\fA$ is the unit of the algebra $\fA$. The map $\omega$ is continuous: $|\omega(\fa)|\leq\|\fa\|$, $\forall\fa\in\fA$. Hence, the state $\omega$ is a positive normed functional over the algebra $\fA$. The set of all states of the algebra $\fA$ will be denoted by $\Omega(\fA)$. The value $\omega(\fa)$ at $\fa=\fa^\ast$ is understood as the average of the observable $\fa$ in the state $\omega$. A general definition of the state of a quantum system can be given in terms of normed density matrices on the Hilbert space $\sH_\infty$. Indeed, any \textit{density matrix} $\rho$ defines the state $\omega_\rho$ on the algebra $\fB(\sH)$ by means of the relation
\[
\omega_\rho(\fa)=\Tr\left[\rho\fa\right],\quad\forall\fa\in\fB(\sH),
\]
which for pure states, $\rho=|\psi\rangle\langle\psi|$, is reduced to the standard expectation $\omega_\rho(\fa)=\langle\psi|\fa|\psi\rangle$.

A set $\Omega(\fA)$ is \textit{convex} if for any states $\omega_1$, $\omega_2$ and $\lambda_1,\,\lambda_2\geq 0$, $\lambda_1+\lambda_2=1$ there is $\lambda_1\omega_1+\lambda_2\omega_2\in\Omega(\fA)$. A state $\omega$ on the $C^\ast$-algebra $\fA$ is called \textit{pure} if it cannot be decomposed into a convex sum of two states, i.e. if the decomposition $\omega=\lambda\omega_1+(1-\lambda)\omega_2$, where $0<\lambda<1$, is performed only for $\omega_1=\omega_2=\omega$. Pure states are extreme points of the set $\Omega(\fA)$. A state $\omega$ that is not pure is called \textit{mixed}.

One of the most important aspects of $C^\ast$-algebra theory is the duality between states and representations of the algebra of observables. The relation between states $\omega$ and irreducible representations $\pi$ of algebras $\fA$ was first explicitly stated by Segal \cite{Seg47}. The canonical correspondence $\omega\leftrightarrow\pi_\omega$ between states and cyclic representations of the $C^\ast$-algebra $\fA$ is given by GNS (Gelfand-Naimark-Segal) construction.
\begin{thm}
{\rm (\textrm{GNS construction} \cite{GN43,Seg47})} For any state $\omega$ (positive functional) on a  $C^\ast$-algebra $\fA$, one can define a cyclic representation $\pi_\omega$ of the algebra $\fA$ in a Hilbert space $\sH$ with a cyclic vector $\left|\Phi\right\rangle$ such that
\[
\omega(\fa)=\langle\Phi\mid\pi_\omega(\fa)\mid\Phi\rangle,\quad\forall\fa\in\fA.
\]
The representation $\pi_\omega$ is defined uniquely by these conditions up to unitary equivalence (correlating cyclic vectors of different representations).
\end{thm}
It follows from the theorem that the concept of a Hilbert space associated with a quantum system is not a primary concept, but an emergent construction, i.e. a consequence of the structure of the $C^\ast$-algebra of the system of observables. Further, each state $\omega$ defines some representation of the algebra $\fA$, and the resulting representation $\pi_\omega$ is irreducible exactly when the state $\omega$ is pure. It is also follows from the theorem that for any non-zero vector $\left|\Phi\right\rangle\in\sH$, the expression
\begin{equation}\label{VectState}
\omega_\Phi(\fa)=\frac{\langle\Phi\mid\pi_\omega(\fa)\mid\Phi\rangle}{\langle\Phi\mid\Phi\rangle},\quad\forall\fa\in\fA,
\end{equation}
defines the state $\omega_\Phi(\fa)$ of the algebra $\fA$, called the \textit{vector state} associated with the representation $\pi_\omega$ and corresponding to the cyclic\footnote{A vector $\left|\Phi\right\rangle\in\sH$ is called cyclic for a representation $\pi$ if all vectors of the form $\pi(\fa)\left|\Phi\right\rangle$ (where $\fa\in\fA$) form a total set in $\sH$, i.e. a set whose closure of the linear shell is everywhere dense in $\sH$. A representation $\pi$ with a cyclic vector is called a cyclic one.} vector $\left|\Phi\right\rangle$. Hence, the set $\Omega_p(\fA)$ of all pure states of the algebra $\fA$ coincides with the set of all vector states associated with all irreducible representations of the algebra $\fA$.

\section{Classification of States}
Separable and entangled (non-separable) states on the $C^\ast$-algebra $\fA$ are the main objects of study in this section. The starting point for constructing such states is the notion of an algebraic bipartition of the operator algebra $\fA$ (see \cite{BF16}).
\begin{defn}
An algebraic bipartition of the $C^\ast$-algebra $\fA$ is any pair $(\fA_1,\fA_2)$ of subalgebras $\fA_1,\fA_2\subset\fA$ such that $\fA_1\cap\fA_2=1_\fA$.
\end{defn}
This directly implies the concept of operator locality.
\begin{defn}
An element of an algebra $\fA$ is called local with respect to a given bipartition $(\fA_1,\fA_2)$ or $(\fA_1,\fA_2)$-local if this element is the product $\fa_1\fa_2$ of the element $\fa_1\in\fA_1$ and the element $\fa_2\in\fA_2$.
\end{defn}
The following definition establishes the most important concepts of separability and entanglement of states on algebra $\fA$.
\begin{defn}
A state $\omega$ on the algebra $\fA$ is called separable with respect to a bipartition $(\fA_1,\fA_2)$ if the expectation $\omega(\fa_1\fa_2)$ of any local operator $\fa_1\fa_2$ can be decomposed into a linear convex combination of products of expectations
\begin{equation}\label{Sep}
\omega(\fa_1\fa_2)=\sum_k\lambda_k\omega^{(1)}_k(\fa_1)\omega^{(2)}_k(\fa_2),\quad\lambda_k\geq 0,\;\;\sum_k\lambda_k=1,
\end{equation}
where $\omega^{(1)}_k$ and $\omega^{(2)}_k$ are states on the algebra $\fA$. Otherwise, the state $\omega$ is called entangled with respect to the bipartition $(\fA_1,\fA_2)$.
\end{defn}
This definition of separability can be easily extended to the case of more than two partitions. For example, for the case of $n$-partition we have
\[
\omega(\fa_1\fa_2\cdots\fa_n)=\sum_k\lambda_k\omega^{(1)}_k(\fa_1)\omega^{(2)}_k(\fa_2)\cdots\omega^{(n)}_k(\fa_n),
\quad\lambda_k\geq 0,\;\;\sum_k\lambda_k=1.
\]
When the state $\omega$ is pure, the separability condition (\ref{Sep}) is simplified.
\begin{defn}
Pure states $\omega$ on the operator algebra $\fA$ are separable with respect to a given bipartition $(\fA_1,\fA_2)$ only and if only
\[
\omega(\fa_1\fa_2)=\omega(\fa_1)\omega(\fa_2)
\]
for all local operators $\fa_1\fa_2$.
\end{defn}
Hence, pure separable states are product states. Taking into account the GNS construction (theorem 1), the general form of any pure separable state  is given by the following theorem.
\begin{thm}{\rm (\textrm{\cite{BF16}})}
Let the state $\omega$ on the algebra $\fA$ be separable with respect to a given bipartition $(\fA_1,\fA_2)$. Then the normed pure state $\left|\psi\right\rangle$ in a GNS-Hilbert space $\sH_\omega$ is $(\fA_1,\fA_2)$-separable if and only if
\[
\left|\psi\right\rangle=\pi_\omega(\fb^{(1)})\pi_\omega(\fb^{(2)})\left|\omega\right\rangle,
\]
where $\fb^{(i)}\in\fA$, $i=1,2$, $\pi_\omega(\fb^{(i)})$ is a cyclic representation of the algebra $\fA$ in the Hilbert space $\sH_\omega$.
\end{thm}
It is obvious that for the case of algebraic $n$-partition, the pure separable state has the form
\[
\left|\psi\right\rangle=\pi_\omega(\fb^{(1)})\pi_\omega(\fb^{(2)})\cdots\pi_\omega(\fb^{(n)})\left|\omega\right\rangle.
\]

\section{Implementation of Operator Algebra}
In this section, we consider a specific implementation of the operator algebra $\fA$. The transition
\[
\fA\Rightarrow\pi(\fA)
\]
from $\fA$ to a specific algebra $\pi(\fA)$, where $\pi$ is the chosen physical representation of the algebra of observables, sometimes called ``dressing'' the operator algebra. Following Heisenberg \cite{Heisen51}, we assume that  at the fundamental level, the main observable is the \textit{\textbf{energy}} that corresponds to the Hermitian operator $H$. As a \textit{fundamental symmetry} that allows us to structure the energy levels of the state spectrum, we choose the group $\SU(2,2)$ (the twofold covering of the conformal group $\SO_0(2,4)$).
\begin{thm}
Let the $C^\ast$-algebra $\fA$ consists of the energy operator $H$ and the generators of the group $\SU(2,2)$ attached to $H$, forming a general system of eigenfunctions with $H$. Then the set $\Omega$ of pure states $\omega$ on the algebra $\fA$ corresponds to a system of cyclic vectors $\left|\psi\right\rangle$ in the GNS-Hilbert space $\sH_\omega$:
\begin{equation}\label{Cycle}
\begin{array}{lcl}
&&\vdots\\
\left|\psi_n\right\rangle&=&\pi_\omega(\fh^{(1)})\pi_\omega(\fh^{(2)})\cdots\pi_\omega(\fh^{(n)}
\left|\omega\right\rangle,\\
\left|\psi_{n-1}\right\rangle&=&\pi_\omega(\fh^{(1)})\pi_\omega(\fh^{(2)})\cdots\pi_\omega(\fh^{(n-1)}
\left|\omega\right\rangle,\\
&&\vdots\\
\left|\psi_2\right\rangle&=&\pi_\omega(\fh^{(1)})\pi_\omega(\fh^{(2)})\left|\omega\right\rangle,\\
\left|\psi_1\right\rangle&=&\pi_\omega(\fh^{(1)})\left|\omega\right\rangle,\\
\left|\psi_0\right\rangle&=&\left|\omega\right\rangle,\\
\left|\psi^\ast_1\right\rangle&=&\pi^\ast_\omega(\fh^{(1)})\left|\omega\right\rangle,\\
\left|\psi^\ast_2\right\rangle&=&\pi^\ast_\omega(\fh^{(1)})\pi^\ast_\omega(\fh^{(2)})\left|\omega\right\rangle,\\
&&\vdots\\
\left|\psi^\ast_{n-1}\right\rangle&=&\pi^\ast_\omega(\fh^{(1)})\pi^\ast_\omega(\fh^{(2)})\cdots\pi^\ast_\omega(\fh^{(n-1)}
\left|\omega\right\rangle,\\
\left|\psi^\ast_n\right\rangle&=&\pi^\ast_\omega(\fh^{(1)})\pi^\ast_\omega(\fh^{(2)})\cdots\pi^\ast_\omega(\fh^{(n)}
\left|\omega\right\rangle,\\
&&\vdots
\end{array}
\end{equation}
where $\fh^{(i)}\in\fA$, $i=1,2,\ldots,n$; $\pi_\omega(\fh^{(i)})$ is a fundamental representation of the spinor group $\spin_+(1,3)$.
\end{thm}
\begin{proof}
Consider the closure $\overline{\fA}$ (the \textit{von Neumann observable algebra}) of an algebra $\fA$ in an $\sigma$-weak operator topology\footnote{According to the von Neumann bicommutant theorem, $\overline{\fA}$ cincides with the repeated commutant $\fA^{cc}$ of the algebra $\fA$ \cite{BLOT}.}. And let the generators of the complex shell of the group algebra $\mathfrak{su}(2,2)$ be self-adjoint operators in $\sH_\infty$ attached to the von Neumann observable algebra $\overline{\fA}$, i.e. such operators that all spectral projectors $E_\lambda$, and hence all bounded functions of the operators belong to $\overline{\fA}$\footnote{It should be noted that already at an early stage of the development of quantum mechanics, in the fundamental work \cite{BHJ26} it was shown that the energy operator $H$ is permuted with all operators in $\sH_\infty$, representing the Lie algebra $\mathfrak{su}(2)$ of the group $\SU(2)$ (see also \cite[p.~138]{RF70}). By virtue of the isomorphism $\mathfrak{sl}(2,\C)\simeq\mathfrak{su}(2)\oplus i\mathfrak{su}(2)$ \cite[p.~28]{Knapp} (the so-called ``unitary trick'' of Weyl), this result can be extended to the group $\SL(2,\C)\simeq\spin_+(1,3)$. Since the operators of group algebra $\mathfrak{sl}(2,\C)$ and the energy operator $H$ commute, then, as a consequence, a general system of eigenfunctions can be constructed for these operators.}.

The twofold covering $\SU(2,2)$ of the conformal group $\SO_0(2,4)$ is isomorphic to the spinor group
\begin{equation}\label{Spin24}
\spin_+(2,4)=\left\{s\in\C_4\;|\;N(s)=1\right\},
\end{equation}
where $\C_4$ is a \textit{Dirac algebra}. On the other hand, twistors can be defined as ``reduced spinors'' of the conformal group $\SO_0(2,4)$. General spinors of the group $\SO_0(2,4)$ are elements of the minimal left ideal of the \textit{conformal algebra} $\cl_{2,4}$:
\[
I_{2,4}=\cl_{2,4}f_{2,4}=\cl_{2,4}\frac{1}{2}(1+\e_{15})\frac{1}{2}(1+\e_{26}),
\]
where $f_{2,4}$ is a primitive idempotent of the algebra $\cl_{2,4}$. Reduced spinors (twistors) are formulated in the framework of an even subalgebra $\cl^+_{2,4}\simeq\cl_{4,1}$ (\textit{de Sitter algebra}). The minimal left ideal of the algebra $\cl_{4,1}\simeq\C_4$ is defined by the following expression:
\[
I_{4,1}=\cl_{4,1}f_{4,1}=\cl_{4,1}\frac{1}{2}(1+\e_0)\frac{1}{2}
(1+i\e_{12}).
\]
Hence, after the reduction $I_{2,4}\rightarrow I_{4,1}$ generated by the isomorphism $\cl^+_{2,4}\simeq\cl_{4,1}$\footnote{The algebra $\cl_{2,4}$ is of $p-q\equiv 6\pmod{8}$, hence, by virtue of the general isomorphism $\cl^+_{p,q}\simeq\cl_{q,p-1}$ (see \cite{Var01}), we have $\cl^+_{2,4}\simeq\cl_{4,1}$, where $\cl_{4,1}$ is the de Sitter algebra associated with the space $\R^{4,1}$. In turn, the algebra $\cl_{4,1}$ is of type $p-q\equiv 3\pmod{8}$, i.e. it has a complex division ring $\K\simeq\C$ and, therefore, has an isomorphism $\cl_{4,1}\simeq\C_4$, where $\C_4$ is a Dirac algebra.}, we see that the twistors $\bsZ^\alpha$ are elements of the ideal $I_{4,1}$, which leads to the group $\SU(2,2)\simeq\spin_+(2,4)\in\cl^+_{2,4}$. Indeed, let us consider the algebra $\cl_{2,4}$ associated with a six-dimensional pseudoeuclidean space $\R^{2,4}$. A twofold covering $\spin_+(2,4)$ of the rotation group $\SO_0(2,4)$ of the space $\R^{2,4}$ is described within the even subalgebra $\cl^+_{2,4}$. The algebra $\cl_{2,4}$ has the type $p-q\equiv 6\pmod{8}$, therefore, according to $\cl^+_{p,q}\simeq\cl_{q,p-1}$ we have $\cl^+_{2,4}\simeq\cl_{4,1}$, where $\cl_{4,1}$ is the de Sitter algebra associated with the space $\R^{4,1}$. In its turn, the algebra $\cl_{4,1}$ has the type $p-q\equiv 3\pmod{8}$ and, therefore, there is an isomorphism $\cl_{4,1}\simeq\C_4$, where $\C_4$ is a Dirac algebra. The algebra $\C_4$ is a comlexification of the space-time algebra: $\C_4\simeq\C\otimes\cl_{1,3}$. Further, $\cl_{1,3}$ admits the following factorization: $\cl_{1,3}\simeq\cl_{1,1}\otimes\cl_{0,2}$. Hence it immediately follows that $\C_4\simeq\C\otimes\cl_{1,1}\otimes\cl_{0,2}$. Thus,
\begin{equation}\label{ConfGroup0}
\spin_+(2,4)=\left\{s\in\C\otimes\cl_{1,1}\otimes\cl_{0,2}\;|\;N(s)=1\right\}.
\end{equation}
On the other hand, in virtue of $\cl_{1,3}\simeq\cl_{1,1}\otimes\cl_{0,2}$ a general element of the algebra $\cl_{1,3}$ can be written in the form
\[
\cA_{\cl_{1,3}}=\cl^0_{1,1}\e_0+\cl^1_{1,1}\phi+\cl^2_{1,1}\psi+\cl^3_{1,1}\phi\psi,
\]
where $\phi=\e_{123}$, $\psi=\e_{124}$ are quaternion units. Therefore,
\begin{equation}\label{ConfGroup}
\spin_+(2,4)=
{\renewcommand{\arraystretch}{1.2}
\left\{s\in\left.\begin{bmatrix} \C\otimes\cl^0_{1,1}-i\C\otimes\cl^3_{1,1} &
-\C\otimes\cl^1_{1,1}+i\C\otimes\cl^2_{1,1}\\
\C\otimes\cl^1_{1,1}+i\C\otimes\cl^2_{1,1} & \C\otimes\cl^0_{1,1}+i\C\otimes\cl^3_{1,1}\end{bmatrix}\right|\;N(s)=1
\right\}.}
\end{equation}
Mappings of the space $\R^{1,3}$, generated by the group $\SO_0(2,4)$, induce linear transformations of the twistor space $\C^4$ with preservation of the form $\bsZ^\alpha\overline{\bsZ}_\alpha$ of the signature $(+,+,-,-)$. Hence it follows that a corresponding group in the twistor space is $\SU(2,2)$ (the group of pseudo-unitary unimodular $4\times 4$ matrices, see (\ref{ConfGroup})):
\begin{equation}\label{ConfGroup2}
\SU(2,2)=\left\{\ar\begin{bmatrix} A & B\\ C & D\end{bmatrix}
\in\C_4:\;\det\begin{bmatrix} A & B \\ C & D\end{bmatrix}=1
\right\}\simeq\spin_+(2,4).
\end{equation}

In turn, the twofold covering $\SL(2,\C)$ of the proper Lorentz group $\SO_0(1,3)$ is isomorphic to the spinor group
\begin{equation}\label{Spin13}
\spin_+(1,3)=\left\{s\in\C_2\;|\;N(s)=1\right\},
\end{equation}
where $\C_2$ is a \textit{biquaternion algebra}. Group (\ref{Spin13}) is a subgroup of group (\ref{Spin24}), $\spin_+(1,3)\subset\spin_+(2,4)$. Hence $\spin_+(2,4)/\spin_+(1,3)$-reduction\footnote{In more detail, the $\spin_+(2,4)/\spin_+(1,3)$-reduction (a reduction based on the Cartan decomposition, see \cite{Var04e,Var06,Var07}) is divided into two consecutive reductions $\SU(2,2)/\Sp(1,1)$ and $\Sp(1,1)/\SL(2,\C)$, where $\Sp(1,1)$ is a twofold covering of the de Sitter group. Here we restrict ourselves to finite-dimensional representations. A more general case of decomposition of infinite-dimensional representations of locally compact groups is considered by Naimark \cite{Nai64}.} of the representation $\fB$ of the group $\spin_+(2,4)$ over the subgroup $\spin_+(1,3)$ leads to the decomposition of $\fB$ into the orthogonal sum of irreducible representations $\pi_i$ of the subgroup $\spin_+(1,3)$:
\[
\fB=\pi_1\oplus\pi_2\oplus\ldots\oplus\pi_i\oplus\ldots
\]
The system of irreducible representations of the group $\spin_+(1,3)\simeq\SL(2,\C)$ is shown in Figure 1.
\begin{figure}[ht]
\unitlength=1.5mm
\begin{center}
\begin{picture}(100,85)(0,-40)
\put(50,0){$\overset{(0,0)}{\bullet}$}\put(47,5.5){\line(1,0){10}}\put(52.25,2.75){\line(0,1){7.25}}
\put(47,-5.5){\line(1,0){10}}\put(52.25,-7.15){\line(0,1){7.25}}
\put(55,5){$\overset{(\frac{1}{2},0)}{\bullet}$}
\put(55,-6){$\overset{(\frac{1}{2},0)}{\bullet}$}
\put(45,5){$\overset{(0,\frac{1}{2})}{\bullet}$}
\put(45,-6){$\overset{(0,\frac{1}{2})}{\bullet}$}
\put(40,10){$\overset{(0,1)}{\bullet}$}\put(42,10.5){\line(1,0){10}}\put(47.25,7.75){\line(0,1){7.25}}
\put(40,-11){$\overset{(0,1)}{\bullet}$}\put(42,-10.5){\line(1,0){10}}\put(47.25,-12.75){\line(0,1){7.25}}
\put(50,10){$\overset{(\frac{1}{2},\frac{1}{2})}{\bullet}$}
\put(50,-11){$\overset{(\frac{1}{2},\frac{1}{2})}{\bullet}$}
\put(52,10.5){\line(1,0){10}}\put(57.25,7.75){\line(0,1){7.25}}
\put(52,-10.5){\line(1,0){10}}\put(57.25,-12.75){\line(0,1){7.25}}
\put(60,10){$\overset{(1,0)}{\bullet}$}
\put(60,-11){$\overset{(1,0)}{\bullet}$}
\put(35,15){$\overset{(0,\frac{3}{2})}{\bullet}$}\put(37,15.5){\line(1,0){10}}\put(42.25,12.75){\line(0,1){7.25}}
\put(35,-16){$\overset{(0,\frac{3}{2})}{\bullet}$}\put(37,-15.5){\line(1,0){10}}\put(42.25,-17.75){\line(0,1){7.25}}
\put(45,15){$\overset{(\frac{1}{2},1)}{\bullet}$}\put(47,15.5){\line(1,0){10}}\put(52.25,12.75){\line(0,1){7.25}}
\put(45,-16){$\overset{(\frac{1}{2},1)}{\bullet}$}\put(47,-15.5){\line(1,0){10}}\put(52.25,-17.75){\line(0,1){7.25}}
\put(55,15){$\overset{(1,\frac{1}{2})}{\bullet}$}\put(57,15.5){\line(1,0){10}}\put(62.25,12.75){\line(0,1){7.25}}
\put(55,-16){$\overset{(1,\frac{1}{2})}{\bullet}$}\put(57,-15.5){\line(1,0){10}}\put(62.25,-17.75){\line(0,1){7.25}}
\put(65,15){$\overset{(\frac{3}{2},0)}{\bullet}$}
\put(65,-16){$\overset{(\frac{3}{2},0)}{\bullet}$}
\put(30,20){$\overset{(0,2)}{\bullet}$}\put(32,20.5){\line(1,0){10}}\put(37.25,17.75){\line(0,1){7.25}}
\put(30,-21){$\overset{(0,2)}{\bullet}$}\put(32,-20.5){\line(1,0){10}}\put(37.25,-22.75){\line(0,1){7.25}}
\put(40,20){$\overset{(\frac{1}{2},\frac{3}{2})}{\bullet}$}
\put(40,-21){$\overset{(\frac{1}{2},\frac{3}{2})}{\bullet}$}
\put(42,20.5){\line(1,0){10}}\put(47.25,17.75){\line(0,1){7.25}}
\put(42,-20.5){\line(1,0){10}}\put(47.25,-22.75){\line(0,1){7.25}}
\put(50,20){$\overset{(1,1)}{\bullet}$}\put(52,20.5){\line(1,0){10}}\put(57.25,17.75){\line(0,1){7.25}}
\put(50,-21){$\overset{(1,1)}{\bullet}$}\put(52,-20.5){\line(1,0){10}}\put(57.25,-22.75){\line(0,1){7.25}}
\put(60,20){$\overset{(\frac{3}{2},\frac{1}{2})}{\bullet}$}
\put(60,-21){$\overset{(\frac{3}{2},\frac{1}{2})}{\bullet}$}
\put(62,20.5){\line(1,0){10}}\put(67.25,17.75){\line(0,1){7.25}}
\put(62,-20.5){\line(1,0){10}}\put(67.25,-22.75){\line(0,1){7.25}}
\put(70,20){$\overset{(2,0)}{\bullet}$}
\put(70,-21){$\overset{(2,0)}{\bullet}$}
\put(25,25){$\overset{(0,\frac{5}{2})}{\bullet}$}\put(27,25.5){\line(1,0){10}}\put(32.25,22.75){\line(0,1){7.25}}
\put(25,-26){$\overset{(0,\frac{5}{2})}{\bullet}$}\put(27,-25.5){\line(1,0){10}}\put(32.25,-27.75){\line(0,1){7.25}}
\put(35,25){$\overset{(\frac{1}{2},2)}{\bullet}$}\put(37,25.5){\line(1,0){10}}\put(42.25,22.75){\line(0,1){7.25}}
\put(35,-26){$\overset{(\frac{1}{2},2)}{\bullet}$}\put(37,-25.5){\line(1,0){10}}\put(42.25,-27.75){\line(0,1){7.25}}
\put(45,25){$\overset{(1,\frac{3}{2})}{\bullet}$}\put(47,25.5){\line(1,0){10}}\put(52.25,22.75){\line(0,1){7.25}}
\put(45,-26){$\overset{(1,\frac{3}{2})}{\bullet}$}\put(47,-25.5){\line(1,0){10}}\put(52.25,-27.75){\line(0,1){7.25}}
\put(55,25){$\overset{(\frac{3}{2},1)}{\bullet}$}\put(57,25.5){\line(1,0){10}}\put(62.25,22.75){\line(0,1){7.25}}
\put(55,-26){$\overset{(\frac{3}{2},1)}{\bullet}$}\put(57,-25.5){\line(1,0){10}}\put(62.25,-27.75){\line(0,1){7.25}}
\put(65,25){$\overset{(2,\frac{1}{2})}{\bullet}$}\put(67,25.5){\line(1,0){10}}\put(72.25,22.75){\line(0,1){7.25}}
\put(65,-26){$\overset{(2,\frac{1}{2})}{\bullet}$}\put(67,-25.5){\line(1,0){10}}\put(72.25,-27.75){\line(0,1){7.25}}
\put(75,25){$\overset{(\frac{5}{2},0)}{\bullet}$}
\put(75,-26){$\overset{(\frac{5}{2},0)}{\bullet}$}
\put(20,30){$\overset{(0,3)}{\bullet}$}\put(22,30.5){\line(1,0){10}}\put(27.25,27.75){\line(0,1){7.25}}
\put(20,-31){$\overset{(0,3)}{\bullet}$}\put(22,-30.5){\line(1,0){10}}\put(27.25,-32.75){\line(0,1){7.25}}
\put(30,30){$\overset{(\frac{1}{2},\frac{5}{2})}{\bullet}$}
\put(30,-31){$\overset{(\frac{1}{2},\frac{5}{2})}{\bullet}$}
\put(32,30.5){\line(1,0){10}}\put(37.25,27.75){\line(0,1){7.25}}
\put(32,-30.5){\line(1,0){10}}\put(37.25,-32.75){\line(0,1){7.25}}
\put(40,30){$\overset{(1,2)}{\bullet}$}\put(42,30.5){\line(1,0){10}}\put(47.25,27.75){\line(0,1){7.25}}
\put(40,-31){$\overset{(1,2)}{\bullet}$}\put(42,-30.5){\line(1,0){10}}\put(47.25,-32.75){\line(0,1){7.25}}
\put(50,30){$\overset{(\frac{3}{2},\frac{3}{2})}{\bullet}$}
\put(50,-31){$\overset{(\frac{3}{2},\frac{3}{2})}{\bullet}$}
\put(52,30.5){\line(1,0){10}}\put(57.25,27.75){\line(0,1){7.25}}
\put(52,-30.5){\line(1,0){10}}\put(57.25,-32.75){\line(0,1){7.25}}
\put(60,30){$\overset{(2,1)}{\bullet}$}\put(62,30.5){\line(1,0){10}}\put(67.25,27.75){\line(0,1){7.25}}
\put(60,-31){$\overset{(2,1)}{\bullet}$}\put(62,-30.5){\line(1,0){10}}\put(67.25,-32.75){\line(0,1){7.25}}
\put(70,30){$\overset{(\frac{5}{2},\frac{5}{2})}{\bullet}$}
\put(70,-31){$\overset{(\frac{5}{2},\frac{5}{2})}{\bullet}$}
\put(72,30.5){\line(1,0){10}}\put(77.25,27.75){\line(0,1){7.25}}
\put(72,-30.5){\line(1,0){10}}\put(77.25,-32.75){\line(0,1){7.25}}
\put(80,30){$\overset{(3,0)}{\bullet}$}
\put(80,-31){$\overset{(3,0)}{\bullet}$}
\put(15,35){$\overset{(0,\frac{7}{2})}{\bullet}$}\put(17,35.5){\line(1,0){10}}\put(22.25,32.75){\line(0,1){7.25}}
\put(15,-36){$\overset{(0,\frac{7}{2})}{\bullet}$}\put(17,-35.5){\line(1,0){10}}\put(22.25,-37.75){\line(0,1){7.25}}
\put(25,35){$\overset{(\frac{1}{2},3)}{\bullet}$}\put(27,35.5){\line(1,0){10}}\put(32.25,32.75){\line(0,1){7.25}}
\put(25,-36){$\overset{(\frac{1}{2},3)}{\bullet}$}\put(27,-35.5){\line(1,0){10}}\put(32.25,-37.75){\line(0,1){7.25}}
\put(35,35){$\overset{(1,\frac{5}{2})}{\bullet}$}\put(37,35.5){\line(1,0){10}}\put(42.25,32.75){\line(0,1){7.25}}
\put(35,-36){$\overset{(1,\frac{5}{2})}{\bullet}$}\put(37,-35.5){\line(1,0){10}}\put(42.25,-37.75){\line(0,1){7.25}}
\put(45,35){$\overset{(\frac{3}{2},2)}{\bullet}$}\put(47,35.5){\line(1,0){10}}\put(52.25,32.75){\line(0,1){7.25}}
\put(45,-36){$\overset{(\frac{3}{2},2)}{\bullet}$}\put(47,-35.5){\line(1,0){10}}\put(52.25,-37.75){\line(0,1){7.25}}
\put(55,35){$\overset{(2,\frac{3}{2})}{\bullet}$}\put(57,35.5){\line(1,0){10}}\put(62.25,32.75){\line(0,1){7.25}}
\put(55,-36){$\overset{(2,\frac{3}{2})}{\bullet}$}\put(57,-35.5){\line(1,0){10}}\put(62.25,-37.75){\line(0,1){7.25}}
\put(65,35){$\overset{(\frac{5}{2},1)}{\bullet}$}\put(67,35.5){\line(1,0){10}}\put(72.25,32.75){\line(0,1){7.25}}
\put(65,-36){$\overset{(\frac{5}{2},1)}{\bullet}$}\put(67,-35.5){\line(1,0){10}}\put(72.25,-37.75){\line(0,1){7.25}}
\put(75,35){$\overset{(3,\frac{1}{2})}{\bullet}$}\put(77,35.5){\line(1,0){10}}\put(82.25,32.75){\line(0,1){7.25}}
\put(75,-36){$\overset{(3,\frac{1}{2})}{\bullet}$}\put(77,-35.5){\line(1,0){10}}\put(82.25,-37.75){\line(0,1){7.25}}
\put(85,35){$\overset{(\frac{7}{2},0)}{\bullet}$}
\put(85,-36){$\overset{(\frac{7}{2},0)}{\bullet}$}
\put(10,40){$\overset{(0,4)}{\bullet}$}\put(12,40.5){\line(1,0){10}}
\put(10,-41){$\overset{(0,4)}{\bullet}$}\put(12,-40.5){\line(1,0){10}}
\put(20,40){$\overset{(\frac{1}{2},\frac{7}{2})}{\bullet}$}\put(22,40.5){\line(1,0){10}}
\put(20,-41){$\overset{(\frac{1}{2},\frac{7}{2})}{\bullet}$}\put(22,-40.5){\line(1,0){10}}
\put(30,40){$\overset{(1,3)}{\bullet}$}\put(32,40.5){\line(1,0){10}}
\put(30,-41){$\overset{(1,3)}{\bullet}$}\put(32,-40.5){\line(1,0){10}}
\put(40,40){$\overset{(\frac{3}{2},\frac{5}{2})}{\bullet}$}\put(42,40.5){\line(1,0){10}}
\put(40,-41){$\overset{(\frac{3}{2},\frac{5}{2})}{\bullet}$}\put(42,-40.5){\line(1,0){10}}
\put(50,40){$\overset{(2,2)}{\bullet}$}\put(52,40.5){\line(1,0){10}}
\put(50,-41){$\overset{(2,2)}{\bullet}$}\put(52,-40.5){\line(1,0){10}}
\put(60,40){$\overset{(\frac{5}{2},\frac{3}{2})}{\bullet}$}\put(62,40.5){\line(1,0){10}}
\put(60,-41){$\overset{(\frac{5}{2},\frac{3}{2})}{\bullet}$}\put(62,-40.5){\line(1,0){10}}
\put(70,40){$\overset{(3,1)}{\bullet}$}\put(72,40.5){\line(1,0){10}}
\put(70,-41){$\overset{(3,1)}{\bullet}$}\put(72,-40.5){\line(1,0){10}}
\put(80,40){$\overset{(\frac{7}{2},\frac{1}{2})}{\bullet}$}\put(82,40.5){\line(1,0){10}}
\put(80,-41){$\overset{(\frac{7}{2},\frac{1}{2})}{\bullet}$}\put(82,-40.5){\line(1,0){10}}
\put(90,40){$\overset{(4,0)}{\bullet}$}
\put(90,-41){$\overset{(4,0)}{\bullet}$}
\put(11.5,45){$\vdots$}
\put(21.5,45){$\vdots$}
\put(31.5,45){$\vdots$}
\put(41.5,45){$\vdots$}
\put(51.5,45){$\vdots$}
\put(61.5,45){$\vdots$}
\put(71.5,45){$\vdots$}
\put(81.5,45){$\vdots$}
\put(91.5,45){$\vdots$}
\put(11.5,-45){$\vdots$}
\put(21.5,-45){$\vdots$}
\put(31.5,-45){$\vdots$}
\put(41.5,-45){$\vdots$}
\put(51.5,-45){$\vdots$}
\put(61.5,-45){$\vdots$}
\put(71.5,-45){$\vdots$}
\put(81.5,-45){$\vdots$}
\put(91.5,-45){$\vdots$}
\put(10,0.5){\line(1,0){42}}\put(50,0.5){\vector(1,0){42}}
\put(16.5,32){$\vdots$}
\put(16.5,29){$\vdots$}
\put(16.5,26){$\vdots$}
\put(16.5,23){$\vdots$}
\put(16.5,20){$\vdots$}
\put(16.5,17){$\vdots$}
\put(16.5,14){$\vdots$}
\put(16.5,11){$\vdots$}
\put(16.5,9){$\vdots$}
\put(16.5,6){$\vdots$}
\put(16.5,3){$\vdots$}
\put(16.5,1.5){$\cdot$}
\put(16.5,0){$\cdot$}
\put(16.5,-32){$\vdots$}
\put(16.5,-29){$\vdots$}
\put(16.5,-26){$\vdots$}
\put(16.5,-23){$\vdots$}
\put(16.5,-20){$\vdots$}
\put(16.5,-17){$\vdots$}
\put(16.5,-14){$\vdots$}
\put(16.5,-11){$\vdots$}
\put(16.5,-9){$\vdots$}
\put(16.5,-6){$\cdot$}
\put(14.5,-3){$-\frac{7}{2}$}
\put(21.5,27){$\vdots$}
\put(21.5,24){$\vdots$}
\put(21.5,21){$\vdots$}
\put(21.5,18){$\vdots$}
\put(21.5,15){$\vdots$}
\put(21.5,13){$\vdots$}
\put(21.5,9){$\vdots$}
\put(21.5,6){$\vdots$}
\put(21.5,3){$\vdots$}
\put(21.5,1.5){$\cdot$}
\put(21.5,0){$\cdot$}
\put(21.5,-27){$\vdots$}
\put(21.5,-24){$\vdots$}
\put(21.5,-21){$\vdots$}
\put(21.5,-18){$\vdots$}
\put(21.5,-15){$\vdots$}
\put(21.5,-13){$\vdots$}
\put(21.5,-9){$\vdots$}
\put(21.5,-6){$\vdots$}
\put(19.5,-3){$-3$}
\put(26.5,22){$\vdots$}
\put(26.5,19){$\vdots$}
\put(26.5,16){$\vdots$}
\put(26.5,13){$\vdots$}
\put(26.5,10){$\vdots$}
\put(26.5,7){$\vdots$}
\put(26.5,4){$\vdots$}
\put(26.5,1){$\vdots$}
\put(26.5,-22){$\vdots$}
\put(26.5,-19){$\vdots$}
\put(26.5,-16){$\vdots$}
\put(26.5,-13){$\vdots$}
\put(26.5,-10){$\vdots$}
\put(26.5,-7){$\vdots$}
\put(26.5,-4){$\vdots$}
\put(26.5,-5){$\cdot$}
\put(24.5,-3){$-\frac{5}{2}$}
\put(31.5,17){$\vdots$}
\put(31.5,14){$\vdots$}
\put(31.5,11){$\vdots$}
\put(31.5,8){$\vdots$}
\put(31.5,5){$\vdots$}
\put(31.5,2){$\vdots$}
\put(31.5,0.5){$\cdot$}
\put(31.5,-17){$\vdots$}
\put(31.5,-14){$\vdots$}
\put(31.5,-11){$\vdots$}
\put(31.5,-8){$\vdots$}
\put(31.5,-5){$\vdots$}
\put(29.5,-3){$-2$}
\put(36.5,12){$\vdots$}
\put(36.5,9){$\vdots$}
\put(36.5,6){$\vdots$}
\put(36.5,3){$\vdots$}
\put(36.5,1.5){$\cdot$}
\put(36.5,0){$\cdot$}
\put(36.5,-12){$\vdots$}
\put(36.5,-9){$\vdots$}
\put(36.5,-6){$\vdots$}
\put(34.5,-3){$-\frac{3}{2}$}
\put(41.5,7){$\vdots$}
\put(41.5,4){$\vdots$}
\put(41.5,1){$\vdots$}
\put(39.5,-3){$-1$}
\put(46.5,2){$\vdots$}
\put(46.5,0.5){$\cdot$}
\put(41.5,-7){$\vdots$}
\put(41.5,-4.5){$\cdot$}
\put(44.5,-2){${\scriptstyle-\frac{1}{2}}$}
\put(51,-3){$0$}
\put(56.5,2){$\vdots$}
\put(56.5,0.5){$\cdot$}
\put(56.5,-2){${\scriptstyle\frac{1}{2}}$}
\put(61.5,7){$\vdots$}
\put(61.5,4){$\vdots$}
\put(61.5,1){$\vdots$}
\put(61.5,-7){$\vdots$}
\put(61.5,-4.5){$\cdot$}
\put(61.5,-3){$1$}
\put(66.5,12){$\vdots$}
\put(66.5,9){$\vdots$}
\put(66.5,6){$\vdots$}
\put(66.5,3){$\vdots$}
\put(66.5,1.5){$\cdot$}
\put(66.5,0){$\cdot$}
\put(66.5,-12){$\vdots$}
\put(66.5,-9){$\vdots$}
\put(66.5,-6){$\vdots$}
\put(66.5,-3){$\frac{3}{2}$}
\put(71.5,17){$\vdots$}
\put(71.5,14){$\vdots$}
\put(71.5,11){$\vdots$}
\put(71.5,8){$\vdots$}
\put(71.5,5){$\vdots$}
\put(71.5,2){$\vdots$}
\put(71.5,0.5){$\cdot$}
\put(71.5,-17){$\vdots$}
\put(71.5,-14){$\vdots$}
\put(71.5,-11){$\vdots$}
\put(71.5,-8){$\vdots$}
\put(71.5,-5){$\vdots$}
\put(71.5,-3){$2$}
\put(76.5,22){$\vdots$}
\put(76.5,19){$\vdots$}
\put(76.5,16){$\vdots$}
\put(76.5,13){$\vdots$}
\put(76.5,10){$\vdots$}
\put(76.5,7){$\vdots$}
\put(76.5,4){$\vdots$}
\put(76.5,1){$\vdots$}
\put(76.5,-22){$\vdots$}
\put(76.5,-19){$\vdots$}
\put(76.5,-16){$\vdots$}
\put(76.5,-13){$\vdots$}
\put(76.5,-10){$\vdots$}
\put(76.5,-7){$\vdots$}
\put(76.5,-5){$\cdot$}
\put(76.5,-3){$\frac{5}{2}$}
\put(81.5,27){$\vdots$}
\put(81.5,24){$\vdots$}
\put(81.5,21){$\vdots$}
\put(81.5,18){$\vdots$}
\put(81.5,15){$\vdots$}
\put(81.5,13){$\vdots$}
\put(81.5,9){$\vdots$}
\put(81.5,6){$\vdots$}
\put(81.5,3){$\vdots$}
\put(81.5,1.5){$\cdot$}
\put(81.5,0){$\cdot$}
\put(81.5,-27){$\vdots$}
\put(81.5,-24){$\vdots$}
\put(81.5,-21){$\vdots$}
\put(81.5,-18){$\vdots$}
\put(81.5,-15){$\vdots$}
\put(81.5,-13){$\vdots$}
\put(81.5,-9){$\vdots$}
\put(81.5,-6){$\vdots$}
\put(81.5,-3){$3$}
\put(86.5,32){$\vdots$}
\put(86.5,29){$\vdots$}
\put(86.5,26){$\vdots$}
\put(86.5,23){$\vdots$}
\put(86.5,20){$\vdots$}
\put(86.5,17){$\vdots$}
\put(86.5,14){$\vdots$}
\put(86.5,11){$\vdots$}
\put(86.5,9){$\vdots$}
\put(86.5,6){$\vdots$}
\put(86.5,3){$\vdots$}
\put(86.5,1.5){$\cdot$}
\put(86.5,0){$\cdot$}
\put(86.5,-32){$\vdots$}
\put(86.5,-29){$\vdots$}
\put(86.5,-26){$\vdots$}
\put(86.5,-23){$\vdots$}
\put(86.5,-20){$\vdots$}
\put(86.5,-17){$\vdots$}
\put(86.5,-14){$\vdots$}
\put(86.5,-11){$\vdots$}
\put(86.5,-9){$\vdots$}
\put(86.5,-6){$\cdot$}
\put(86.5,-3){$\frac{7}{2}$}
\put(53.8,1.7){$\cdot$}\put(54.3,2.2){$\cdot$}\put(54.8,2.7){$\cdot$}\put(55.3,3.3){$\cdot$}\put(55.8,3.8){$\cdot$}
\put(56.3,4.3){$\cdot$}
\put(52.8,-0.7){$\cdot$}\put(53.3,-1.2){$\cdot$}\put(53.8,-1.7){$\cdot$}\put(54.3,-2.3){$\cdot$}\put(54.8,-2.8){$\cdot$}
\put(55.3,-3.3){$\cdot$}
\put(58.8,6.8){$\cdot$}\put(59.3,7.3){$\cdot$}\put(59.8,7.8){$\cdot$}\put(60.3,8.3){$\cdot$}\put(60.8,8.8){$\cdot$}
\put(61.3,9.3){$\cdot$}
\put(57.8,-5.8){$\cdot$}\put(58.3,-6.3){$\cdot$}\put(58.8,-6.8){$\cdot$}\put(59.3,-7.3){$\cdot$}\put(59.8,-7.8){$\cdot$}
\put(60.3,-8.3){$\cdot$}
\put(63.8,11.8){$\cdot$}\put(64.3,12.3){$\cdot$}\put(64.8,12.8){$\cdot$}\put(65.3,13.3){$\cdot$}\put(65.8,13.8){$\cdot$}
\put(66.3,14.3){$\cdot$}
\put(62.8,-10.8){$\cdot$}\put(63.3,-11.3){$\cdot$}\put(63.8,-11.8){$\cdot$}\put(64.3,-12.3){$\cdot$}\put(64.8,-12.8){$\cdot$}
\put(65.3,-13.3){$\cdot$}
\put(68.8,16.8){$\cdot$}\put(69.3,17.3){$\cdot$}\put(69.8,17.8){$\cdot$}\put(70.3,18.3){$\cdot$}\put(70.8,18.8){$\cdot$}
\put(71.3,19.3){$\cdot$}
\put(67.8,-15.8){$\cdot$}\put(68.3,-16.3){$\cdot$}\put(68.8,-16.8){$\cdot$}\put(69.3,-17.3){$\cdot$}\put(69.8,-17.8){$\cdot$}
\put(70.3,-18.3){$\cdot$}
\put(33.8,21.8){$\cdot$}\put(34.3,22.3){$\cdot$}\put(34.8,22.8){$\cdot$}\put(35.3,23.3){$\cdot$}\put(35.8,23.8){$\cdot$}
\put(36.3,24.3){$\cdot$}
\put(32.3,-21.8){$\cdot$}\put(32.8,-22.3){$\cdot$}\put(33.3,-22.8){$\cdot$}\put(33.8,-23.3){$\cdot$}\put(34.3,-23.8){$\cdot$}
\put(34.8,-24.3){$\cdot$}
\put(38.8,26.8){$\cdot$}\put(39.3,27.3){$\cdot$}\put(39.8,27.8){$\cdot$}\put(40.3,28.3){$\cdot$}\put(40.8,28.8){$\cdot$}
\put(41.3,29.3){$\cdot$}
\put(37.3,-26.8){$\cdot$}\put(37.8,-27.3){$\cdot$}\put(38.3,-27.8){$\cdot$}\put(38.8,-28.3){$\cdot$}\put(39.3,-28.8){$\cdot$}
\put(39.8,-29.3){$\cdot$}
\put(43.8,31.8){$\cdot$}\put(44.3,32.3){$\cdot$}\put(44.8,32.8){$\cdot$}\put(45.3,33.3){$\cdot$}\put(45.8,33.8){$\cdot$}
\put(46.3,34.3){$\cdot$}
\put(42.3,-31.8){$\cdot$}\put(42.8,-32.3){$\cdot$}\put(43.3,-32.8){$\cdot$}\put(43.8,-33.3){$\cdot$}\put(44.3,-33.8){$\cdot$}
\put(44.8,-34.3){$\cdot$}
\put(48.8,36.8){$\cdot$}\put(49.3,37.3){$\cdot$}\put(49.8,37.8){$\cdot$}\put(50.3,38.3){$\cdot$}\put(50.8,38.8){$\cdot$}
\put(51.3,39.3){$\cdot$}
\put(47.3,-36.8){$\cdot$}\put(47.8,-37.3){$\cdot$}\put(48.3,-37.8){$\cdot$}\put(48.8,-38.3){$\cdot$}\put(49.3,-38.8){$\cdot$}
\put(49.8,-39.3){$\cdot$}
\put(47.3,4.4){$\cdot$}\put(47.8,3.9){$\cdot$}\put(48.3,3.4){$\cdot$}\put(48.8,2.9){$\cdot$}\put(49.3,2.4){$\cdot$}
\put(49.8,1.9){$\cdot$}
\put(50.8,-0.7){$\cdot$}\put(50.3,-1.2){$\cdot$}\put(49.8,-1.7){$\cdot$}\put(49.3,-2.2){$\cdot$}\put(48.8,-2.7){$\cdot$}
\put(48.3,-3.2){$\cdot$}
\put(42.3,9.4){$\cdot$}\put(42.8,8.9){$\cdot$}\put(43.3,8.4){$\cdot$}\put(43.8,7.9){$\cdot$}\put(44.3,7.4){$\cdot$}
\put(44.8,6.9){$\cdot$}
\put(45.8,-5.8){$\cdot$}\put(45.3,-6.3){$\cdot$}\put(44.8,-6.8){$\cdot$}\put(44.3,-7.3){$\cdot$}\put(43.8,-7.8){$\cdot$}
\put(43.3,-8.3){$\cdot$}
\put(37.3,14.4){$\cdot$}\put(37.8,13.9){$\cdot$}\put(38.3,13.4){$\cdot$}\put(38.8,12.9){$\cdot$}\put(39.3,12.4){$\cdot$}
\put(39.8,11.9){$\cdot$}
\put(40.8,-10.8){$\cdot$}\put(40.3,-11.3){$\cdot$}\put(39.8,-11.8){$\cdot$}\put(39.3,-12.3){$\cdot$}\put(38.8,-12.8){$\cdot$}
\put(38.3,-13.3){$\cdot$}
\put(32.3,19.4){$\cdot$}\put(32.8,18.9){$\cdot$}\put(33.3,18.4){$\cdot$}\put(33.8,17.9){$\cdot$}\put(34.3,17.4){$\cdot$}
\put(34.8,16.9){$\cdot$}
\put(35.8,-15.8){$\cdot$}\put(35.3,-16.3){$\cdot$}\put(34.8,-16.8){$\cdot$}\put(34.3,-17.3){$\cdot$}\put(33.8,-17.8){$\cdot$}
\put(33.3,-18.3){$\cdot$}
\put(67.3,24.4){$\cdot$}\put(67.8,23.9){$\cdot$}\put(68.3,23.4){$\cdot$}\put(68.8,22.9){$\cdot$}\put(69.3,22.4){$\cdot$}
\put(69.8,21.9){$\cdot$}
\put(71.3,-21.8){$\cdot$}\put(70.8,-22.3){$\cdot$}\put(70.3,-22.8){$\cdot$}\put(69.8,-23.3){$\cdot$}\put(69.3,-23.8){$\cdot$}
\put(68.8,-24.3){$\cdot$}
\put(62.3,29.4){$\cdot$}\put(62.8,28.9){$\cdot$}\put(63.3,28.4){$\cdot$}\put(63.8,27.9){$\cdot$}\put(64.3,27.4){$\cdot$}
\put(64.8,26.9){$\cdot$}
\put(66.3,-26.8){$\cdot$}\put(65.8,-27.3){$\cdot$}\put(65.3,-27.8){$\cdot$}\put(64.8,-28.3){$\cdot$}\put(64.3,-28.8){$\cdot$}
\put(63.8,-29.3){$\cdot$}
\put(57.3,34.4){$\cdot$}\put(57.8,33.9){$\cdot$}\put(58.3,33.4){$\cdot$}\put(58.8,32.9){$\cdot$}\put(59.3,32.4){$\cdot$}
\put(59.8,31.9){$\cdot$}
\put(61.3,-31.8){$\cdot$}\put(60.8,-32.3){$\cdot$}\put(60.3,-32.8){$\cdot$}\put(59.8,-33.3){$\cdot$}\put(59.3,-33.8){$\cdot$}
\put(58.8,-34.3){$\cdot$}
\put(52.3,39.4){$\cdot$}\put(52.8,38.9){$\cdot$}\put(53.3,38.4){$\cdot$}\put(53.8,37.9){$\cdot$}\put(54.3,37.4){$\cdot$}
\put(54.8,36.9){$\cdot$}
\put(56.3,-36.8){$\cdot$}\put(55.8,-37.3){$\cdot$}\put(55.3,-37.8){$\cdot$}\put(54.8,-38.3){$\cdot$}\put(54.3,-38.8){$\cdot$}
\put(53.8,-39.3){$\cdot$}
\end{picture}
\end{center}
\vspace{0.3cm}
\begin{center}\begin{minipage}{30pc}{\small {\bf Figure 1:} The system of cyclic representations $(l,\dot{l})$ of the Lorentz group, where $l=k/2$, $\dot{l}=r/2$ (cone of representations). The values of the spin lines are marked on the axis.}\end{minipage}\end{center}
\end{figure}

Consider an arbitrary proper subspace $E_\lambda$ of the energy operator $H$. As noted above, the operators $\sX_l$, $\sY_l$ of the complex shell\footnote{The operators $\sX_l=1/2i(\sA_l+i\sB_l)$, $\sY_l=1/2i(\sA_l-i\sB_l)$ ($l=1,2,3$) of the complex shell of the group algebra $\mathfrak{sl}(2,\C)\simeq\mathfrak{su}(2)\oplus i\mathfrak{su}(2)$ form the Van der Waerden representation \cite{Wa32} of the Lorentz group, which is related to the Gelfand-Naimark representation operators $F$, $G$ \cite{GMS} by the following formulas: $\sX_\pm=1/2(F_\pm-iH_\pm)$, $\sY_\pm=-1/2(F_\pm+iH_\pm)$, $\sX_3=1/2(F_3-iH_3)$, $\sY_3=-1/2(F_3+iH_3)$.} of the group algebra $\mathfrak{sl}(2,\C)$ and the energy operator $H$ commute with each other, therefore, a general system of eigenfunctions can be constructed for these operators. This means that the subspace $E_\lambda$ is invariant with respect to the operators $\sX_l$, $\sY_l$ (moreover, the operators $\sX_l$, $\sY_l$ can \textit{only} be considered on $E_\lambda$). This allows us to identify subspaces $E_\lambda$ with symmetric spaces $\Sym_{(k,r)}$ of interlocking (entangling) representations $\boldsymbol{\tau}_{k/2,r/2}$ of the Lorentz group\footnote{Finite-dimensional representations $\boldsymbol{\tau}_{l\dot{l}}$ of the Lorentz group (usually denoted as $\fD^{(j/2,r/2)}$, here $j/2=l$, $r/2=\dot{l}$) play an important role in the axiomatic quantum field theory \cite{BLOT,SW64}. For more details about interlocking representations $\boldsymbol{\tau}_{l\dot{l}}$ see \cite{GY48,AD72,PS83}.} and thereby obtain a concrete implementation (``dressing'') of the operator algebra $\pi(\fA)\Rightarrow\pi(H)$, where the cyclic representation is $\pi\equiv\boldsymbol{\tau}_{k/2,r/2}$. The structure of the representation $\boldsymbol{\tau}_{k/2,r/2}$ is defined by the following tensor product:
\begin{equation}\label{TenRep}
\boldsymbol{\tau}_{\frac{k}{2},\frac{r}{2}}=\boldsymbol{\tau}_{\frac{k}{2},0}\otimes
\boldsymbol{\tau}_{0,\frac{r}{2}}\simeq
\underbrace{\boldsymbol{\tau}_{\frac{1}{2},0}\otimes\boldsymbol{\tau}_{\frac{1}{2},0}\otimes\cdots\otimes
\boldsymbol{\tau}_{\frac{1}{2},0}}_{k\;\text{times}}\bigotimes
\underbrace{\boldsymbol{\tau}_{0,\frac{1}{2}}\otimes\boldsymbol{\tau}_{0,\frac{1}{2}}\otimes\cdots\otimes
\boldsymbol{\tau}_{0,\frac{1}{2}}}_{r\;\text{times}},
\end{equation}
where $\boldsymbol{\tau}_{\frac{1}{2},0}$($\boldsymbol{\tau}_{0,\frac{1}{2}}$) is the fundamental representation of the group $\SL(2,\C)$. Then the map
\[
\pi_\omega(\fh^{(1)})\pi_\omega(\fh^{(2)})\cdots\pi_\omega(\fh^{(n)}
\left|\omega\right\rangle\;\longmapsto\;\boldsymbol{\tau}_{\frac{k}{2},0}\otimes
\boldsymbol{\tau}_{0,\frac{r}{2}}\left|\omega\right\rangle
\]
defines a system of cyclic vectors $\left|\psi\right\rangle$ in the GNS-Hilbert space $\sH_\omega$, with $\left|\psi_0\right\rangle=\boldsymbol{\tau}_{0,0}\left|\omega\right\rangle$, where $\boldsymbol{\tau}_{0,0}$ is the unit representation of the group $\SL(2,\C)\simeq\spin_+(1,3)$. Hence, the pure separable states $\omega$ of the operator algebra $\pi(\fA)$ (the energy operator $H$) correspond to the cyclic vectors (\ref{Cycle}) in $\sH_\omega$.
\end{proof}
Within the binary structure implemented by irreducible representations of the spinor group $\spin_+(1,3)$, we have pure separable states of arbitrary spin (see Figure 1). The states of the half-integer spin form \textit{\textbf{fermionic lines}} $1/2$, $3/2$, $\ldots$ (respectively, dual fermionic lines $-1/2$, $-3/2$, $\ldots$). The states of the integer spin form \textit{\textbf{bosonic lines}} $0$, $1$, $2$, $\ldots$ (respectively, dual bosonic lines $-1$, $-2$, $\ldots$). In this case, the fermionic lines correspond to cyclic vectors $\boldsymbol{\tau}_{\frac{k}{2},0}\otimes\boldsymbol{\tau}_{0,\frac{r}{2}}\left|\omega\right\rangle$ with an odd number of factors $\boldsymbol{\tau}_{\frac{1}{2},0}$ ($\boldsymbol{\tau}_{0,\frac{1}{2}}$) in the tensor product, the bosonic lines correspond to cyclic vectors with an even tensor product. Hence, $\dZ_2$-grading naturally occurs here. Further, the Clifford algebra $\cl$ is associated with each cyclic vector. In the case of a number field $\F=\C$, $\dZ_2$-grading leads to a supergroup $G(\C_n,\gamma,\bcirc)$ implementing Cartan-Bott periodicity for $\C_n$ algebras, where the cyclic action of the supergroup is given by the Brauer-Wall group $BW_\C\simeq\dZ_2$ (for more details, see \cite{Var12,Var15c}).

The simplest spin $1/2$ fermionic states, belonging to the representation cone (see Figure 1) form the following quadruplet:
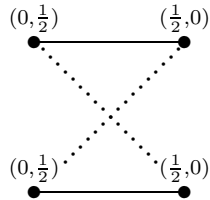
\unitlength=1mm
\begin{center}
\begin{picture}(20,20)(0,-11)
\put(0,5){$\overset{(0,\frac{1}{2})}{\bullet}$}
\put(20,5){$\overset{(\frac{1}{2},0)}{\bullet}$}
\put(0,-15){$\overset{(0,\frac{1}{2})}{\bullet}$}
\put(20,-15){$\overset{(\frac{1}{2},0)}{\bullet}$}
\put(3,6){\line(1,0){20}}
\put(3,-14){\line(1,0){20}}
\put(4,4){$\cdot$}
\put(5,3){$\cdot$}
\put(6,2){$\cdot$}
\put(7,1){$\cdot$}
\put(8,0){$\cdot$}
\put(9,-1){$\cdot$}
\put(10,-2){$\cdot$}
\put(11,-3){$\cdot$}
\put(12,-4){$\cdot$}
\put(13,-5){$\cdot$}
\put(14,-6){$\cdot$}
\put(15,-7){$\cdot$}
\put(16,-8){$\cdot$}
\put(17,-9){$\cdot$}
\put(18,-10){$\cdot$}
\put(19,-11){$\cdot$}
\put(22,4){$\cdot$}
\put(21,3){$\cdot$}
\put(20,2){$\cdot$}
\put(19,1){$\cdot$}
\put(18,0){$\cdot$}
\put(17,-1){$\cdot$}
\put(16,-2){$\cdot$}
\put(15,-3){$\cdot$}
\put(14,-4){$\cdot$}
\put(13,-5){$\cdot$}
\put(12,-6){$\cdot$}
\put(11,-7){$\cdot$}
\put(10,-8){$\cdot$}
\put(9,-9){$\cdot$}
\put(8,-10){$\cdot$}
\put(7,-11){$\cdot$}
\end{picture}
\end{center}
Since these states are formed by fundamental representations $\boldsymbol{\tau}_{\frac{1}{2},0}$ ($\boldsymbol{\tau}_{0,\frac{1}{2}}$), it is natural to assume that they correspond to an electron, but not as a particle\footnote{Rumer and Fet in the book \cite{RF70} write: ``Until now, we believed that \textit{the same particle}, for example, an electron, can be in two spin states with the spin +1/2 and -1/2. However, an electron without a certain spin value is never observed and is only an abstract concept. In view of this, another point of view is quite natural: we can assume that there are \textit{two} elementary particles -- an electron with spin +1/2 and an electron with spin -1/2, while ``just an electron'' does not occur in nature. At the same time, it is possible to preserve the concept of an electron as an abstract particle whose energy in a magnetic field is always split into two possible values'' \cite[pp.~161-162]{RF70}. Next, in Dirac we find: ``In the quantum theory, however, discontinuous transitions may take place, so that if the electron is initially in a state of positive kinetic energy it may make a transition to a state of negative kinetic energy'' \cite[p.~273]{Dir}. Thus, a certain ``abstract particle'' can be in four states: as an electron in two spin states -1/2 or +1/2, and as a positron in two spin states -1/2 or 1/2. In this example, it becomes clearly clear that the classical concept of a particle is an absolutely abstract (alien) concept at the micro-level. The electron as a particle is a fiction that exists only in human consciousness. As Erich Joos noted: ``There are no particles'' (www.decoherence.de), see also \cite{deRonde}. There are no ``particles'', there are only states.}.

\section{$\K$-Hilbert Space}
Quantum mechanics can be defined using Hilbert spaces over any of three associative division algebras: the real numbers $\R$, the complex numbers $\C$, and the quaternions $\BH$. Jordan, von Neumann, and Wigner \cite{JNW34} in their historically first classification of observable algebras considered real, complex, and quaternion quantum theories on equal grounds. However, unlike the generally accepted complex version, the real and quaternion versions of quantum mechanics face a number of problems (the absence of Stone's theorem, the tensor product of two quaternion Hilbert spaces is not a quaternion Hilbert space). It is shown in \cite{Baez} that these problems can be solved if we interpret the real, complex, and quaternion quantum theories as parts of a single structure. This structure is also known as the Freemen Dyson ``threefold way'' \cite{Dys62} (see also Arnold's mathematical ``trinities'' \cite{Arn99}).

According to Hurwitz's theorem \cite{Hur}, there are four normed division algebras: $\R$, $\C$, $\BH$ and $\dO$. Unlike the first three, the octonion algebra $\dO$ is nonassociative. For quantum mechanics, it is important that every normed division algebra is a $C^\ast$-algebra in which a real linear map
\begin{eqnarray}
\fA&\longrightarrow&\fA\nonumber\\
x&\longmapsto&x^\ast,\nonumber
\end{eqnarray}
is defined that satisfies the equalities $(xy)^\ast=y^\ast x^\ast$, $(x^\ast)^\ast=x$. For $\C$, this map is a  \textit{complex conjugate}; for $\R$, it is an identity map. For quaternions $\BH$,
 \[
(a\boldsymbol{1}+b\mathbf{i}+c\mathbf{j}+d\mathbf{k})^\ast=a\mathbf{1}-b\mathbf{i}-c\mathbf{j}-d\mathbf{k}
\]
holds\footnote{Similarly, for octonions $\dO$ the $\ast$-operation is given by the expression
\[
(a_0\mathbf{1}+a_1\mathbf{e}_1+\ldots+a_7\mathbf{e}_7)^\ast=a_0\mathbf{1}-a_1\mathbf{e}_1-\ldots-a_7\mathbf{e}_7.
\]
}.
In all cases ($\R$, $\C$, $\BH$ and $\dO$), the identity $xx^\ast=x^\ast x=\|x\|^2\mathbf{1}$ holds.

For three associative division algebras ($\R$, $\C$, $\BH$), the structure of the $C^\ast$-algebra allows us to define a Hilbert space. Let $\K$ be an associative normed division algebra. Then a $\K$-\textit{\textbf{vector space}} will be a right $\K$-module if there exists such an Abelian group $V$ equipped with a map
\begin{eqnarray}
V\times\K&\longrightarrow&V\nonumber\\
(v,X)&\longmapsto&vx,\nonumber
\end{eqnarray}
satisfying the conditions
\[
(v+w)(x)=vx+wx,\quad v(x+y)=vx+vy,\quad (vx)y=v(xy).
\]
In this case, the mapping $T:\;V\rightarrow V^\prime$ between $\K$-vector spaces is $\K$-\textit{\textbf{linear}} if
\[
T(vx+wy)=T(v)x+T(w)y
\]
for all $v,w\in V$ and $x,y\in\K$. Further, the \textit{\textbf{inner product}} on a $\K$-vector space is defined by a map
\begin{eqnarray}
V\times V&\longrightarrow&\K\nonumber\\
(v,w)&\longmapsto&\langle v\mid w\rangle,\nonumber
\end{eqnarray}
that is $\K$-linear in the second argument and satisfies the equality $\langle v\mid w\rangle=\langle w\mid v\rangle^\ast$ for all $v,w\in V$, and is also \textit{\textbf{positive-definite}}: $(v,v)\geq 0$. As usual, the inner product defines the norm $\|v\|=\sqrt{\langle v\mid v\rangle}$, which translates the vector space into a complete metric space. Then a $\K$-vector space with an inner product is a $\K$-\textit{\textbf{Hilbert space}}.
\begin{thm}
Let the $C^\ast$-algebra $\fA$ consists of the energy operator $H$ and the generators of the group $\SU(2,2)$ attached to $H$, forming a general system of eigenfunctions with $H$. And let the set of pure separable states $\omega$ on the algebra $\fA$ correspond to a system of cyclic vectors $\left|\psi\right\rangle$ in the GNS-Hilbert space $\sH_\omega$. Then the $\K$-linear structure ($\K=\R,\C,\BH$) translates $\sH_\omega$ into a physical $\K$-Hilbert space $\bsH_{\rm phys}(\K)$, in which three basis sectors are allocated:\\
1) Charged state sector $\bsH_{\rm phys}(\C)$.\\
2) Neutral state sector $\bsH_{\rm phys}(\BH)$.\\
3) Truly neutral state sector $\bsH_{\rm phys}(\R)$.
\end{thm}
\begin{proof}
According to theorem 3, the pure separable state $\omega$ of an operator algebra $\pi(\fA)$ is given by a cyclic vector $\boldsymbol{\tau}_{\frac{k}{2},0}\otimes\boldsymbol{\tau}_{0,\frac{r}{2}}\left|\omega\right\rangle$ in the GNS-Hilbert space $\sH_\omega$. Representation (\ref{TenRep}) acts in the symmetric space
\[
\mbox{\small Sym}_{(k,r)}\!\!=\!\underbrace{\mbox{\small Sym}_{(1,0)}\otimes\mbox{\small Sym}_{(1,0)}\otimes\cdots\otimes\mbox{\small Sym}_{(1,0)}}_{k\;\text{times}}\bigotimes
\underbrace{\mbox{\small Sym}_{(0,1)}\otimes\mbox{\small Sym}_{(0,1)}\otimes\cdots\otimes\mbox{\small Sym}_{(0,1)}}_{r\;\text{times}}.
\]
Vectors of the space $\Sym_{(k,r)}$ are \textit{symmetric spintensors} resulting from the operation of symmetrization of the general spintensors of the form
\[
\boldsymbol{S}=\boldsymbol{s}^{\alpha_1\alpha_2\ldots\alpha_k\dot{\alpha}_1\dot{\alpha}_2\ldots
\dot{\alpha}_r}=\sum \boldsymbol{s}^{\alpha_1}\otimes
\boldsymbol{s}^{\alpha_2}\otimes\cdots\otimes
\boldsymbol{s}^{\alpha_k}\otimes
\boldsymbol{s}^{\dot{\alpha}_1}\otimes
\boldsymbol{s}^{\dot{\alpha}_2}\otimes\cdots\otimes
\boldsymbol{s}^{\dot{\alpha}_r},
\]
which, in turn, are vectors of spinspace
\[
\dS_{2^{k+r}}=\dS_{2^k}\otimes\dot{\dS}_{2^r}\simeq
\underbrace{\dS_2\otimes\dS_2\otimes\cdots\otimes\dS_2}_{k\;\text{times}}\bigotimes
\underbrace{\dot{\dS}_2\otimes\dot{\dS}_2\otimes\cdots\otimes\dot{\dS}_2}_{r\;\text{times}}.
\]
Any pair of substitutions
\[
\alpha=\begin{pmatrix} 1 & 2 & \ldots & k\\
\alpha_1 & \alpha_2 & \ldots & \alpha_k\end{pmatrix},\quad \beta=\begin{pmatrix} 1 & 2 & \ldots & r\\
\dot{\alpha}_1 & \dot{\alpha}_2 & \ldots & \dot{\alpha}_r\end{pmatrix}
\]
defines a transformation $(\alpha,\beta)$ that maps $\boldsymbol{S}$ to the following polynomial:
\[
P_{\alpha\beta}\boldsymbol{S}=\boldsymbol{s}^{\alpha\left(\alpha_1\right)\alpha\left(\alpha_2\right)\ldots
\alpha\left(\alpha_k\right)\beta\left(\dot{\alpha}_1\right)\beta\left(\dot{\alpha}_2\right)\ldots
\beta\left(\dot{\alpha}_r\right)}.
\]
A spintensor $\boldsymbol{S}$ is called a \emph{symmetric spintensor} if the equality
$
P_{\alpha\beta}\boldsymbol{S}=\boldsymbol{S}
$
holds for any $\alpha$, $\beta$. The space $\Sym_{(k,r)}$ of symmetric spintensors has dimension
$
\dim\Sym_{(k,r)}=(k+1)(r+1)
$.
The dimension of space $\Sym_{(k,r)}$ is called the \emph{degree of representation} $\boldsymbol{\tau}_{l\dot{l}}$ of the group $\SL(2,\C)$. It is easy to see that $\SL(2,\C)$ has representations of \textbf{\emph{any degree}} (as opposed to $\SU(3)$, $\SU(6)$ and other groups of internal symmetries).

Further, the spinspace $\dS_{2^{n/2}}$ is the minimal left ideal $I_{p,q}$ of the Clifford algebra $\cl_{p,q}$\footnote{This definition follows from the algebraic theory of Chevalley \cite{Che54}. By virtue of the Wedderburn-Artin theorem, for a Clifford algebra $\cl_{p,q}$ over a field $\F=\R$, there is an isomorphism $\cl_{p,q}\simeq\End_{\K}(I_{p,q})\simeq\Mat_{2^m}(\K)$, where $m=\frac{p+q}{2}$, $I_{p,q}=\cl_{p,q}f$ is the minimal left ideal of the algebra $\cl_{p,q}$, and $\K=f\cl_{p,q}f$ is the division ring for $\cl_{p,q}$. Hence, if $\cl_{p,q}$ is a simple algebra, then the map $\cl_{p,q}\overset{\gamma}{\longrightarrow}\End_{\K}(\dS)$,
$u\longrightarrow\gamma(u)$, $\gamma(u)\psi=u\psi$ gives an irreducible representation of the algebra $\cl_{p,q}$ in the spinspace $\dS_{2^m}(\K)\simeq I_{p,q}=\cl_{p,q}f$, where $\psi\in\dS_{2^m}$, $m=\frac{p+q}{2}$. On the other hand, if $\cl_{p,q}$ is a semisimple algebra, then the map $\cl_{p,q}\overset{\gamma}{\longrightarrow}\End_{\K\oplus\hat{\K}}(\dS\oplus\hat{\dS})$, $u\longrightarrow\gamma(u)$, $\gamma(u)\psi=u\psi$ gives a reducible representation of the algebra $\cl_{p,q}$ in the double spinspace $\dS\oplus\hat{\dS}$, where $\hat{\dS}=\{\hat{\psi}|\psi\in\dS\}$. In this case, the ideal $\dS\oplus\hat{\dS}$ has a right $\K\oplus\hat{\K}$-linear structure, $\hat{\K}=\{\hat{\lambda}|\lambda\in\K\}$, and $\K\oplus\hat{\K}$ is isomorphic to a double real division ring $\R\oplus\R$, if $p-q\equiv 1\pmod{8}$, or a double quaternion division ring $\BH\oplus\BH$ if $p-q\equiv 5\pmod{8}$. The map $\gamma$ defines the so-called {\it left-regular} spinor representation of the algebra $\cl(Q)$ in the spinspaces $\dS$ and $\dS\oplus\hat{\dS}$, respectively.
}:
\[
\dS_{2^{n/2}}=I_{p,q}=\cl_{p,q}f_{p,q},
\]
where $f_{p,q}$ is the primitive idempotent of the algebra $\cl_{p,q}$, $n=p+q$. Clifford algebras $\cl_{p,q}$ over the field $\F=\R$ are subdivided into eight different types with the following ring structure.\\[0.3cm]
{\bf I}. Central-simple algebras.
\begin{description}
\item[1] Two types $p-q\equiv 0,2\pmod{8}$ with a real division ring $\K\simeq\R$.
\item[2] Two types $p-q\equiv 3,7\pmod{8}$ with a complex division ring $\K\simeq\C$.
\item[3] Two types $p-q\equiv 4,6\pmod{8}$ with a quaternion division ring $\K\simeq\BH$.
\end{description}
{\bf II}. Semisimple algebras.
\begin{description}
\item[4] Type $p-q\equiv 1\pmod{8}$ with a double real division ring $\K\simeq\R\oplus\R$.
\item[5] Type $p-q\equiv 5\pmod{8}$ with a double quaternion division ring $\K\simeq\BH\oplus\BH$.
\end{description}
Accordingly, a $\K$- or $\K\oplus\hat{\K}$-linear structure\footnote{In the case of odd-dimensional algebras $\cl_{p,q}$ with rings $\K\oplus\K$ (types $p-q\equiv 1,5\pmod{8}$), there are homomorphic maps $\epsilon:\,\cl_{p,q}\rightarrow\cl_{p,q-1}$, $\epsilon:\,\cl_{p,q}\rightarrow\cl_{q,p-1}$, where the quotient algebras have the form
${}^\epsilon\cl_{p,q-1}\simeq\cl_{p,q}/\Ker\epsilon$, ${}^\epsilon\cl_{q,p-1}\simeq\cl_{p,q}/\Ker\epsilon$,
$\Ker\epsilon=\left\{\cA^1-\omega\cA^1\right\}$ is the kernel of the homomorphism $\epsilon$. In this case, the double ideals $\dS\oplus\hat{\dS}$ can be replaced with quotient ideals ${}^\epsilon\dS$ and go to quotient representations ${}^\epsilon\pi$ (for more details, see \cite{Var01,Var04}).} is associated with each cyclic vector $\left|\psi\right\rangle$ of the space $\sH_\omega$, depending on the tensor dimension of the representation $\pi\equiv\boldsymbol{\tau}_{k/2,r/2}$. The presence of a $\K$-linear structure translates the GNS-Hilbert space $\sH_\omega$ into a $\K$-Hilbert space\footnote{According to theorem 3, the space $\sH_\omega$ is an emergent construction, i.e. a consequence of the structure of the $C^\ast$-algebra $\fA$ of observables. Similarly, the $\K$($\K\oplus\K$)-linear structure in $\sH_\omega$ is not pre-defined (as in the usual definition), but appears in an emergent way depending on the tensor structure of cyclic vectors.}.

Since the spinspace $\dS_{2^{k+r}}(\K)$ is associated with the cyclic representation $\pi\equiv\boldsymbol{\tau}_{k/2,r/2}$, we have three types of representations: real $\pi_\R$, complex $\pi_\C$ and quaternion $\pi_\BH$ cyclic representations. In this case, there are the following $\K\otimes\K$-transitions:
\begin{eqnarray}
\dS_{2^{n_1}}(\R)\otimes\dS_{2^{n_2}}(\R)&\simeq&\dS_{2^{n_1+n_2}}(\R):\;\;\R\otimes\R\rightarrow\R,\nonumber\\
\dS_{2^{n_1}}(\R)\otimes\dS_{2^{n_2}}(\BH)&\simeq&\dS_{2^{n_1+n_2}}(\BH):\;\;\R\otimes\BH\rightarrow\BH,\nonumber\\
\dS_{2^{n_1}}(\BH)\otimes\dS_{2^{n_2}}(\R)&\simeq&\dS_{2^{n_1+n_2}}(\BH):\;\;\BH\otimes\R\rightarrow\BH,\nonumber\\
\dS_{2^{n_1}}(\BH)\otimes\dS_{2^{n_2}}(\BH)&\simeq&\dS_{2^{n_1+n_2}}(\R):\;\;\BH\otimes\BH\rightarrow\R,\nonumber\\
\dS_{2^{n_1}}(\C)\otimes\dS_{2^{n_2}}(\R)&\simeq&\dS_{2^{n_1+n_2}}(\C):\;\;\C\otimes\R\rightarrow\C,\nonumber\\
\dS_{2^{n_1}}(\R)\otimes\dS_{2^{n_2}}(\C)&\simeq&\dS_{2^{n_1+n_2}}(\C):\;\;\R\otimes\C\rightarrow\C,\label{Trans}\\
\dS_{2^{n_1}}(\C)\otimes\dS_{2^{n_2}}(\BH)&\simeq&\dS_{2^{n_1+n_2}}(\C):\;\;\C\otimes\BH\rightarrow\C,\nonumber\\
\dS_{2^{n_1}}(\BH)\otimes\dS_{2^{n_2}}(\C)&\simeq&\dS_{2^{n_1+n_2}}(\C):\;\;\BH\otimes\C\rightarrow\C,\nonumber\\
\dS_{2^{n_1}}(\C)\otimes\dS_{2^{n_2}}(\C)&\simeq&\dS_{2^{n_1+n_2}}(\C):\;\;\C\otimes\C\rightarrow\C,\nonumber\\
\dS_{2^{n_1}}(\C)\otimes\dS_{2^{n_2}}(\overline{\C})&\simeq&\dS_{2^{n_1+n_2}}(\R):
\;\;\C\otimes\overline{\C}\rightarrow\R,\nonumber\\
\dS_{2^{n_1}}(\BH)\otimes\dS_{2^{n_2}}(\overline{\BH})&\simeq&\dS_{2^{n_1+n_2}}(\R):
\;\;\BH\otimes\overline{\BH}\rightarrow\R.\nonumber
\end{eqnarray}
Further, each complex Clifford algebra $\C_n=\C\otimes\cl_{p,q}$ is associated with a complex vector space $\C^n$, where $n=p+q$. The operation of selecting a real subspace $\R^{p,q}$ in the space $\C^n$ is the basis for determining the discrete transformation, known in physics as the \textit{\textbf{charge conjugation}} $C$. Any element $\cA\in\C_n$ can be uniquely represented as $\cA=\cA_1+i\cA_2$, where $\cA_1,\,\cA_2\in\cl_{p,q}$. Then the map
\[
\cA\longrightarrow\overline{\cA}=\cA_1-i\cA_2
\]
translates the algebra $\C_n$ into itself one-to-one and preserving the operations of addition and multiplication of elements $\cA$ (the operation of multiplying an element by a number passes into the operation of multiplying by a complex-conjugate number). Any mapping of an algebra $\C_n$ to itself that has the listed properties is called a {\it pseudo-automorphism} \cite{Rash}. Thus, the allocation of $\R^{p,q}$ in $\C^n$ induces pseudo-automorphism $\cA\rightarrow\overline{\cA}$ in $\C_n$. Spinor representations of the pseudo-automorphism $\cA\rightarrow\overline{\cA}$ are defined by theorem 1 in \cite{Var01} (see also \cite{Var04}).

Cyclic vectors $\left|\psi\right\rangle$ of a $\K$-Hilbert space $\bsH_{\rm phys}(\K)$ with an associated $\C$- or $\C\oplus\C$-structure define \textit{\textbf{charged states}}. Cyclic vectors $\left|\psi\right\rangle$ with a complex-conjugate $\overline{\C}$- or $\overline{\C}\oplus\overline{\C}$-structure correspond to antimatter. Further, cyclic vectors $\left|\psi\right\rangle$ with an associated $\BH$- or $\BH\oplus\BH$-structure define \textit{\textbf{neutral states}}\footnote{Since the real spinor structure appears as a result of reduction $\C_{2(k+r)}\rightarrow\cl_{p,q}$, then, as a consequence, the charge conjugation $C$ (pseudo-automorphism $\cA\rightarrow\overline{\cA}$) for algebras $\cl_{p,q}$ over the real number field $\F=\R$ and the quaternion division ring $\K\simeq\BH$ (types $p-q\equiv 4,6\pmod{8}$) is reduced to the \textit{state-antistate exchange} $C^\prime$ (theorem 1 in \cite{Var01}, see also \cite{Var04,Var14}).}. As in the case of $\K\simeq\C$, cyclic vectors $\left|\psi\right\rangle$ with a conjugate $\overline{\BH}$- or $\overline{\BH}\oplus\overline{\BH}$-structure correspond to antimatter. In turn, cyclic vectors $\left|\psi\right\rangle$ with an associated $\R$- or $\R\oplus\R$-structure define \textit{\textbf{truly neutral states}}\footnote{In the case of truly neutral states, the pseudo-automorphism $\cA\rightarrow\overline{\cA}$ is reduced to the identity transformation (the state coincides with its anti-state).}.

Thus, the physical $\K$-Hilbert space $\bsH_{\rm phys}(\K)$ is divided into three sectors: $\bsH_{\rm phys}(\C)$, $\bsH_{\rm phys}(\BH)$ and $\bsH_{\rm phys}(\R)$. However, the number of states is not fixed in any of the three sectors. According to (\ref{Trans}), states from one sector move to another, thus forming a single structure (similar to the Dyson ``threefold way''\footnote{According to \cite{Baez}, the Dyson ``threefold way'' \cite{Dys62} can be implemented within the framework of category theory. Let $\Hilb_\R$, $\Hilb_\C$, and $\Hilb_\BH$ be categories of real, complex, and quaternion Hilbert spaces, respectively. Then there are functors that translate the categories $\Hilb_\R$ and $\Hilb_\BH$ to $\Hilb_\C$. And there are also functors that map $\Hilb_\C$ and $\Hilb_\BH$ to $\Hilb_\R$ (respectively, $\Hilb_\R$ and $\Hilb_\C$ to $\Hilb_\BH$). It follows that none of the three forms of quantum mechanics is ``privileged'': each contains the other two.}). Within a single ternary structure, each cyclic vector $\left|\psi\right\rangle\in\bsH_{\rm phys}(\K)$ has a tensor structure (energy, mass) and a $\K$-linear structure (charge), and the combination of these two structures leads to a dynamic change in charge and mass.
\end{proof}

\section{Coherent Subspaces}
In this section, we will continue to study the structure of the physical $\K$-Hilbert space $\bsH_{\rm phys}(\K)$.

As is known, Hilbert spaces of any objects can be represented by a subspace of the space of the tensor product of two-dimensional Hilbert spaces:
\[
\sH_m\subseteq\bsT_n=\underset{n}{\bigotimes}\bsH_2,\quad m\leq 2^n.
\]
According to theorem 3, the fundamental cyclic representation $\pi\equiv\boldsymbol{\tau}_{1/2,0}$ ($\boldsymbol{\tau}_{0,1/2}$) of the algebra $\fA$ of observables (the energy operator $H$ and the generators of the group $\SU(2,2)$ attached to $H$) is defined in a two-dimensional GNS-Hilbert space $\sH_2$. Further, it follows from theorem 4 that the fundamental cyclic representations $\pi_\K$, where $\K=\R,\,\C,\,\BH$, are in turn defined in two-dimensional $\K$-Hilbert spaces $\bsH_2(\K)$. The fundamental pure separable states of the $C^\ast$-algebra $\fA$ of observables are given by cyclic vectors $\left|\psi\right\rangle$ of spaces $\bsH_2(\K)$. The vectors $\left|\psi\right\rangle$ of representations $\pi_\K$ are associated with two-dimensional spinspaces $\dS_2(\K)$, which are left minimal ideals of algebras
\begin{eqnarray}
\cl_{0,2}&\simeq&\BH,\;\;\;\;\quad p-q\equiv 6\pmod{8}\;-\;\text{quaternion algebra};\nonumber\\
\cl_{2,0}&\simeq&\R(2),\quad p-q\equiv 2\pmod{8}\;-\;\text{anti-quaternion algebra};\nonumber\\
\cl_{1,1}&\simeq&\R(2),\quad p-q\equiv 0\pmod{8}\;-\;\text{pseudo-quaternion algebra}\nonumber
\end{eqnarray}
for the case of a number field $\F=\R$\footnote{Here we use Rosenfeld's terminology \cite{Roz95}. In \cite{Roz95}, it is shown that there is a wide variety of different geometries over three associative division algebras: $\R$, $\C$, and $\BH$.}. In the case of $\F=\C$, we come to the biquaternion algebra $\C_2$, which is a complexification of the algebras $\cl_{0,2}$, $\cl_{2,0}$, and $\cl_{1,1}$\footnote{The quaternion algebras $\cl_{0,2}$, $\cl_{2,0}$, and $\cl_{1,1}$ are the simplest ``building blocks'' of a binary structure, since the algebras $\cl_{0,1}$ and $\cl_{1,0}$ are isomorphic to number fields. So, the algebra $\cl_{0,1}$ has the type $p-q\equiv 7\pmod{8}$, hence $\cl_{0,1}\simeq\cl_{0,0}\oplus\omega\cl_{0,0}$, where $\cl_{0,0}\simeq\R$ is a field of real numbers, and the element $\omega=\mathbf{e}_1$ by virtue of $\omega^2=-1$ belongs to the center $\bZ_{0,1}$ of the algebra $\cl_{0,1}$. Thus, $\cl_{0,1}\simeq\R\oplus i\R=\C$ is a field of complex numbers. Similarly, for the algebra $\cl_{1,0}$ of type $p-q\equiv 1\pmod{8}$ the isomorphism $\cl_{1,0}\simeq\cl_{0,0}\oplus\omega\cl_{0,0}\simeq\R\oplus e\R={}^2\R$ is valid, where ${}^2\R$ is a field of \textit{double numbers}, $\omega\equiv e$, and $e$ is a double unit, $e^2=1$.}:
\begin{eqnarray}
\C_2&\simeq&\C\otimes\cl_{0,2},\nonumber\\
&\simeq&\C\otimes\cl_{2,0},\nonumber\\
&\simeq&\C\otimes\cl_{1,1}.\nonumber
\end{eqnarray}
In addition, there is an isomorphism\footnote{The structure of this isomorphism is defined as follows: the maximal basis element $\omega=\mathbf{e}_1\mathbf{e}_2\mathbf{e}_3$ of the algebra $\cl_{3,0}$ commutes with all the basis elements of this algebra, hence $\omega$ belongs to the center $\bZ_{3,0}=\{1,\omega\}$ of the algebra $\cl_{3,0}$. Since $\omega^2=-1$, $\bZ_{3,0}$ is isomorphic to the field of complex numbers, $\bZ_{3,0}=\{1,i\}\simeq\C$, which implies the isomorphism $\cl_{3,0}\simeq\C_2$.} $\C_2\simeq\cl_{3,0}$.

Let $\left|\psi\right\rangle$ be a cyclic vector of the space $\bsH_{\rm phys}(\K)$ (according to theorem 3, the vector $\left|\psi\right\rangle$ sets the pure separable state $\omega$ of the $C^\ast$-algebra $\fA$ of observables). Then $\boldsymbol{\Psi}=e^{i\alpha}\left|\psi\right\rangle$, where $\alpha$ runs through all the real numbers and $\sqrt{\left\langle\psi\right.\!\left|\psi\right\rangle}=1$, will be called a \emph{unit ray}. Hence, the unit ray  $\boldsymbol{\Psi}$ is a collection of basis cyclic vectors $\{\lambda\left|\psi\right\rangle\}$, $\lambda=e^{i\alpha}$, $\left|\psi\right\rangle\in\bsH_{\rm phys}(\K)$. The values associated with the observed effects are absolute values of the semi-linear form $|\left\langle\psi_1\right.\!\left|\psi_2\right\rangle|^2$, independent of the parameters $\lambda$ that characterize the ray. Thus, the \textit{\textbf{ray space}} is quotient space $\hat{H}=\bsH_{\rm phys}(\K)/S^1$, that is, the projective space of one-dimensional subspaces from $\bsH_{\rm phys}(\K)$. All states of a single quantum system $\bsU$ (the spectrum of matter) are described by unit rays. Assume that the basic correspondence between physical states and elements of space $\bsH_{\rm phys}(\K)$ includes the \textit{superposition principle} of quantum theory, i.e. there is a set of basis states such that arbitrary states can be constructed from them using linear superpositions. However, as is known \cite{WWW52}, not all unit rays are physically realizable. There are physical restrictions (\textit{\textbf{superselection rules}}) on the implementation of the superposition principle. In 1952, Wigner, Wightman, and Wick \cite{WWW52} showed that the existence of superselection rules is related to the measurability of the relative phase of a superposition. This means that a pure state cannot be realized as a superposition of certain states, for example, there is no pure state (coherent superposition) of bosonic $\left|\Psi_b\right\rangle$ and fermionic $\left|\Psi_f\right\rangle$ states (the spin superselection rules). However, if a density matrix $\rho$ is defined in $\bsH_{\rm phys}(\K)$, then the superposition $\left|\Psi_b\right\rangle+\left|\Psi_f\right\rangle$ defines a mixed state.
\begin{thm}
A physical $\K$-Hilbert space $\bsH_{\rm phys}(\K)$ admits decomposition into a direct sum of (non-zero) coherent subspaces
\[
\bsH_{\rm phys}(\K)=\underset{\nu\in N}{\bigoplus}\bsH^\nu_{\rm phys}(\K).
\]
In this case, the superposition principle takes place in a restricted form, i.e. within coherent subspaces $\bsH^\nu_{\rm phys}(\K)$. A non-zero linear combination of cyclic vectors of pure separable states is a cyclic vector of a pure separable state, provided that all the original vectors lie in the same coherent subspace $\bsH^\nu_{\rm phys}(\K)$. A superposition of cyclic vectors of pure separable states from various coherent subspaces defines a mixed state.
\end{thm}
\begin{proof}
The starting point of the proof is the correspondence $\omega\leftrightarrow\left|\psi\right\rangle$ between the pure separable states of the operator algebra $\fA$ and the cyclic vectors of the space $\bsH_{\rm phys}(\K)$ (theorems 3 and 4). Let $PS(\fA)$ be the set of all pure states of the $C^\ast$-algebra $\fA$ of observables and let $\left|\psi_1\right\rangle,\,\left|\psi_2\right\rangle\in\bsH_{\rm phys}(\K)$. Then, as we know \cite{RR69}, for an arbitrary $C^\ast$-algebra $\fA$, the probability of transition between two pure states $\omega_1,\,\omega_2\in PS(\fA)$ is determined by the Roberts-Roepstorff formula
\[
|\langle\psi_1\mid\psi_2\rangle|^2=\omega_1\cdot\omega_2=1-1/4\|\omega_1-\omega_2\|^2.
\]
In this case, $\omega_1\cdot\omega_2=\omega_2\cdot\omega_1$ and $\omega_1\cdot\omega_2$ is always enclosed in the segment $[0,1]$. Accordingly, $\omega_1\cdot\omega_2=1$ exactly when $\omega_1=\omega_2$. We will call two pure states $\omega_1$ and $\omega_2$ \textit{orthogonal} if the probability of transition $\omega_1\cdot\omega_2$ is zero. Accordingly, two subsets $S_1$ and $S_2$ in $PS(\fA)$ are mutually orthogonal if $\omega_1\cdot\omega_2=0$ for all $\omega_1\in S_1$ and $\omega_2\in S_2$. Further, a non-empty subset $S\in PS(\fA)$ is called \textit{non-decaying} if it cannot be split into two non-empty orthogonal subsets. Following \cite{BLOT}, we assume that every maximal non-decaying set (i.e., a non-decaying set that is not a proper subset of another non-decaying set of pure states from $PS(\fA)$) is a \textit{sector}. So, $PS(\fA)$ is divided into sectors, therefore, in $PS(\fA)$ there is an equivalence relation $\omega_1\sim\omega_2$ exactly when there is a non-decaying set in $PS(\fA)$ containing $\omega_1$ and $\omega_2$. Hence, $PS(\fA)$ is uniquely divided into pairwise disjoint and mutually orthogonal sectors that exactly coincide with the \textit{equivalence classes} in $PS(\fA)$. In fact, a sector is an algebraic counterpart of a coherent subspace. Further, let a set of cyclic vectors in a physical $\K$-Hilbert space $\bsH_{\rm phys}(\K)$, corresponding to the pure separable states of the algebra $\fA$ of observables (in our case, the energy operator $H$ and the generators of the group $\SU(2,2)$ attached to $H$), forms a total set in $\bsH_{\rm phys}(\K)$, i.e. such set $X$ whose closure of the linear shell is everywhere dense in $\bsH_{\rm phys}(\K)$. Then $X$ cannot be represented as a union of two (or more) non-empty mutually orthogonal subsets. We will say that the vectors $\left|\psi_1\right\rangle$, $\left|\psi_2\right\rangle\in X$ are connected by the relation $\left|\psi_1\right\rangle\sim\left|\psi_2\right\rangle$ if $\left|\psi_1\right\rangle$ and $\left|\psi_2\right\rangle$ belong to the linear shell of $X$. It is easy to see that the relation $\left|\psi_1\right\rangle\sim\left|\psi_2\right\rangle$ is induced by the equivalence relation $\omega_1\sim\omega_2$ from $PS(\fA)$. Hence, the relation $\left|\psi_1\right\rangle\sim\left|\psi_2\right\rangle$ is an equivalence relation and the equivalence classes in $X$ form a partition of $X$ into mutually orthogonal systems $X_\nu$, where $\{\nu\}=N$ is a certain family of indices. Taking now as $\bsH^\nu_{\rm phys}(\K)$ the closed linear shell of the set $X_\nu$, we arrive at the desired decomposition of $\bsH_{\rm phys}(\K)$ into a direct sum of mutually orthogonal subspaces $\bsH^\nu_{\rm phys}(\K)$. This implies a restricted form of the superposition principle (namely, within subspaces $\bsH^\nu_{\rm phys}(\K)$).
\end{proof}

Further, let $G=U(1)^n\equiv U(1)\times\ldots\times U(1)$ be a connected compact $n$-parametric Abelian group (\textit{gauge group}) defined in $\bsH_{\rm phys}(\K)$ by a central extension\footnote{The group $\bcE=\{(\omega,x)\}$ is called the \emph{central extension} of the group $G=\{(e,x)\}$ by means of the group $\bcK=\{(\omega,e)\}$, where $\bcK$ is an Abelian group formed by multiplying the nonequivalent phases $\omega$. Vector representations of group $\bcE$ contain all ray representations of group $G$ (for more information about the central extension technique, see \cite{BR77}).}. An arbitrary element of this group is represented by a set of $n$ phase factors:
\[
g(s_1,\ldots,s_n)\equiv(e^{i\alpha_1},\ldots,e^{i\alpha_n}),\quad 0\leq\alpha_j<2\pi.
\]
In the space $\bsH_{\rm phys}(\K)$, we define an exact unitary representation $U$ of the group $G$. The corresponding gauge transformations in $\bsH_{\rm phys}(\K)$ have the form
\[
U(g)=s^{Q_1}_1\ldots s^{Q_n}_n\equiv e^{i(\alpha_1Q_1+\ldots+\alpha_nQ_n)}.
\]
Generators $Q_1$, $\ldots$, $Q_n$ of gauge transformations are mutually commuting self-adjoint operators with an integer spectrum (they are called \textit{charges} corresponding to a given gauge group). Then, in accordance with Theorem 5, $\bsH_{\rm phys}(\K)$ is decomposed into a direct sum
\begin{equation}\label{SRules}
\bsH_{\rm phys}(\K)=\bigoplus_{q_1,\ldots,q_n\in\dZ}\bsH^{(q_1,\ldots,q_n)}_{\rm phys}(\K)
\end{equation}
of the corresponding spectral subspaces consisting of all vectors $\left|\psi\right\rangle$ such that $(Q_j-q_j)\left|\psi\right\rangle=0$, $j=1,\ldots,n$. In this case, an arbitrary non-zero cyclic vector $\left|\psi\right\rangle\in\bsH_{\rm phys}(\K)$ determines the pure separable state $\omega$ of the algebra of observables exactly when it is an eigenvector for all charges. In this way, the \textit{standard superselection rules} are given in $\bsH_{\rm phys}(\K)$, and (\ref{SRules}) is a decomposition of $\bsH_{\rm phys}(\K)$ into a direct sum of coherent subspaces $\bsH^{(q_1,\ldots,q_n)}_{\rm phys}(\K)$. According to the current state of high-energy physics, the superselection rules can be quite fully described by electric $Q$ ($=Q_1$), baryon $B$ ($=Q_2$), and lepton $L$ ($=Q_3$) charges. According to theorem 4, the electric charge is taken into account by the $\K$-linear structure of the space $\bsH_{\rm phys}(\K)$. In order to describe the entire spectrum of observed states (levels of the matter spectrum), we introduce a 2-parameter gauge group $G=U(1)^2\equiv U(1)\times U(1)$ with respect to the baryon $B$ and lepton $L$ charges. Then the decomposition of $\bsH_{\rm phys}(\K)$ into coherent subspaces has the form:
\begin{equation}\label{KS}
\bsH_{\rm phys}(\K)=\bigoplus_{b,\ell\in\dZ}\bsH^{(b,\ell)}_{\rm phys}(\K),\quad\K=\R,\C,\BH,
\end{equation}
where $b$ and $\ell$ are the baryon and lepton numbers. Consequently, the entire set of cyclic vectors of the space $\bsH_{\rm phys}(\K)$ is divided into subspaces of vectors of the form
\begin{equation}\label{Vector}
\left|\K,b,\ell,s\right\rangle=\left|\K,b,\ell,|l-\dot{l}|\right\rangle
\end{equation}
with given values of charge, spin, masss, baryon and lepton numbers. In addition, a corresponding $CPT$ group is associated with each vector (\ref{Vector}), see \cite{Var04}.

For two vectors $\left|\K,b_1,\ell_1,s_1\right\rangle$ and $\left|\K,b_2,\ell_2,s_2\right\rangle$ of the form (\ref{Vector}), we define the \textit{\textbf{fusion}} operation:
\[
\left|\K,b_1,\ell_1,s_1\right\rangle\otimes\left|\K,b_2,\ell_2,s_2\right\rangle=
\left|\K\otimes\K,b_1+b_2,\ell_1+\ell_2,s_1+s_2\right\rangle,
\]
where for the tensor product $\K\otimes\K$ we have the relations (\ref{Trans}). Consider the fusion operation for the simplest case of two neutrino states $\left|\nu\right\rangle=\left|\BH,0,1,1/2\right\rangle$ (neutrino) and $\left|\bar{\nu}\right\rangle=\left|\overline{\BH},0,-1,1/2\right\rangle$ (antineutrino\footnote{The lepton number $\ell$ for antistates is negative ($\ell=-1$). The baryon number $b$ for leptons is equal to zero. Obviously, the relation (\ref{Fusion}) is valid for any type of neutrino ($\nu_e$, $\nu_\mu$, $\nu_\tau$), since regardless of the type (flavor), each neutrino belongs to a coherent subspace $\bsH^{(0,1)}_{\rm phys}(\BH)$.}):
\begin{equation}\label{Fusion}
\left|\nu\right\rangle\otimes\left|\bar{\nu}\right\rangle=\left|\BH,0,1,\frac{1}{2}\right\rangle\otimes
\left|\overline{\BH},0,-1,\frac{1}{2}\right\rangle=\left|\BH\otimes\overline{\BH},0,0,1\right\rangle=
\left|\R,0,0,1\right\rangle=\left|\gamma\right\rangle,
\end{equation}
where $\left|\gamma\right\rangle$ is the photon\footnote{The fusion $\left|\nu\right\rangle\otimes\left|\bar{\nu}\right\rangle=\left|\gamma\right\rangle$ leads to the \textit{neutrino theory of light}, formulated by de Broglie in 1932 \cite{Bro32} (long before the experimental discovery of neutrinos in 1956), in which the photon is represented as a bound state of neutrino and antineutrino. De Broglie proposed the existence of primary spin 1/2 particles, which he called ``corpuscles'', union (fusion) which allows us to obtain a particle of any spin (fusion theory \cite{Bro}). This idea was later used by Bargmann and Wigner \cite{BW48} to construct irreducible representations of the Lorentz group and the corresponding relativistic equations for particles with arbitrary spin (the Bargmann-Wigner formalism). Heisenberg's non-linear unified spinor theory of matter \cite{Heisen3} is also adjacent to the same direction. Thus, the construction of a unified theory covering all particles and fields originates from the idea of de Broglie to base on the simplest wave function of the spinor type, describing a particle of a minimal spin $s=1/2$. This direction is known in theoretical physics as ``spinorism''. In a sence, de Broglie corpuscles are prototypes of quarks, since in the quark model, the formation of hadrons (in the framework of $q\bar{q}$- and $qqq$-schemes) is represented as an analog of the fusion (confinement) of corpuscules-quarks (fundamental particles of spin 1/2).}. The state $\left|\gamma\right\rangle=\left|\R,0,0,1\right\rangle$ belongs to the coherent subspace $\bsH^{(0,0)}_{\rm phys}(\R)$. In formula (\ref{Fusion}), the isomorphism $\BH\otimes\overline{\BH}\simeq\R$ (resp. $\BH\otimes\BH\simeq\R$) plays a decisive role.

\section{Fusion, Doubling, and Annihilation of States}
Within the framework of binary structure, there are three basic operations for states with minimal tensor dimension:\\
1) \textit{\textbf{Fusion}}.
\[
\left|\mathfrak{q}\right\rangle\otimes\left|\bar{\mathfrak{q}}\right\rangle=\left|\BH,0,1,\frac{1}{2}\right\rangle\otimes
\left|\overline{\BH},0,-1,\frac{1}{2}\right\rangle=\left|\BH\otimes\overline{\BH},0,0,1\right\rangle=
\left|\R,0,0,1\right\rangle=\left|\gamma\right\rangle.
\]
2) \textit{\textbf{Doubling}}.
\begin{eqnarray}
\left|e^-\right\rangle&=&\left|\mathfrak{q}\right\rangle\oplus\left|\bar{\mathfrak{q}}\right\rangle=
\left|\BH\oplus i\BH,0,1,\frac{1}{2}\right\rangle=\left|\C,0,1,\frac{1}{2}\right\rangle,\nonumber\\
\left|e^+\right\rangle&=&\left|\mathfrak{q}\right\rangle\ominus\left|\bar{\mathfrak{q}}\right\rangle=
\left|\BH\ominus i\BH,0,-1,\frac{1}{2}\right\rangle=\left|\overline{\C},0,-1,\frac{1}{2}\right\rangle.\nonumber
\end{eqnarray}
3) \textit{\textbf{Annihilation}}.
\begin{multline}
\left|e^-\right\rangle\otimes\left|e^+\right\rangle=\left|\C,0,1,\frac{1}{2}\right\rangle\otimes
\left|\overline{\C},0,-1,\frac{1}{2}\right\rangle=\\
=\left|\BH\oplus i\BH,0,1,\frac{1}{2}\right\rangle\otimes\left|\BH\ominus i\BH,0,-1,\frac{1}{2}\right\rangle=
\left|\BH\otimes\BH,0,0,1\right\rangle+i\left|\BH\otimes\BH,0,0,1\right\rangle-\\
-i\left|\BH\otimes\BH,0,0,1\right\rangle+\left|\BH\otimes\BH,0,0,1\right\rangle=2\left|\R,0,0,1\right\rangle=
2\left|\gamma\right\rangle.\nonumber
\end{multline}
The latter operation is an algebraic analog of the process of annihilation of an electron-positron pair into two $\gamma$-quanta: $e^-e^+\rightarrow 2\gamma$.

We show that all the cyclic vectors of the physical $\K$-Hilbert space, which represent the pure separable states of the matter spectrum, can be determined by the above-mentioned fusion and doubling operations. To this end, we will need to recall some information about the factorization and periodicity of Clifford algebras.

Let $\cl(V,Q)$ be a Clifford algebra over a field $\F=\R$, where $V$ is a vector space equipped with a quadratic form $Q=x^2_1+\ldots+x^2_p-\ldots-x^2_{p+q}$. If $p+q$ is even and $\omega^2=1$, then $\cl(V,Q)$ is called {\it positive} and, respectively, {\it negative} if $\omega^2=-1$, i.e. $\cl_{p,q}>0$ if $p-q\equiv 0,4\pmod{8}$ and $\cl_{p,q}<0$ if $p-q\equiv 2,6\pmod{8}$.
\begin{thm}[{\rm Karoubi \cite{Karo}}]\label{tKaroubi}
1) If $\cl(V,Q)>0$ and $\dim V$ is even, then
\[
\cl(V\oplus V^{\prime},Q\oplus
Q^{\prime})\simeq\cl(V,Q)\otimes\cl(V^{\prime},Q^{\prime}).
\]
2) If $\cl(V,Q)<0$ and $\dim V$ is even, then
\[
\cl(V\oplus V^{\prime},Q\oplus
Q^{\prime})\simeq\cl(V,Q)\otimes\cl(V^{\prime},-Q^{\prime}).
\]
\end{thm}
Using Karoubi's theorem, we obtain the following factorization for the algebra $\cl_{p,q}$ ($p-q\equiv 0\pmod{2}$):
\begin{equation}\label{Factor}
\cl_{p,q}\simeq\underbrace{\cl_{s_i,t_j}\otimes\cl_{s_i,t_j}\otimes\cdots
\otimes\cl_{s_i,t_j}}_{(p+q)/2\;\text{times}},
\end{equation}
where $s_i,t_j\in\{0,1,2\}$.
\begin{thm} All cyclic vectors $\left|\psi\right\rangle$ of the physical $\K$-Hilbert space $\bsH_{\rm phys}(\K)$, which define the fermionic and bosonic states of $\R$-, $\C$-, and $\BH$-subspaces, are determined by a composition of fusion and doubling operations from the active $\left|\mathfrak{q}_a\right\rangle$ and inert $\left|\mathfrak{q}_s\right\rangle$ fundamental states.
\end{thm}
\begin{proof}
According to Theorem 3, the \textit{fermionic states} $F$ correspond to cyclic vectors $\boldsymbol{\tau}_{\frac{k}{2},0}\otimes\boldsymbol{\tau}_{0,\frac{r}{2}}\left|\omega\right\rangle$ with an odd number of cofactors  $\boldsymbol{\tau}_{\frac{1}{2},0}$ (resp. $\boldsymbol{\tau}_{0,\frac{1}{2}}$) in the tensor product. In turn, the \textit{bosonic states} $B$ correspond to cyclic vectors with an even number of cofactors. Therefore,
\begin{equation}\label{State}
\underbrace{\boldsymbol{\tau}_{\frac{1}{2},0}\otimes\boldsymbol{\tau}_{\frac{1}{2},0}\otimes\cdots\otimes
\boldsymbol{\tau}_{\frac{1}{2},0}\bigotimes
\boldsymbol{\tau}_{0,\frac{1}{2}}\otimes\boldsymbol{\tau}_{0,\frac{1}{2}}\otimes\cdots\otimes
\boldsymbol{\tau}_{0,\frac{1}{2}}}_{m\;\text{times}}\Rightarrow\left\{\begin{array}{lc}
F,& m\equiv 1\pmod{2};\\
B, & m\equiv 0\pmod{2}.
\end{array}\right.
\end{equation}
The two-dimensional Hilbert space $\bsH_2(\C)$, in which the fundamental cyclic representation $\boldsymbol{\tau}_{\frac{1}{2},0}$ (resp. $\boldsymbol{\tau}_{0,\frac{1}{2}}$), i.e. the ``basic building blocks'' of the structure (\ref{State}), is equivalent to the spinspace $\dS_2(\C)$, which, in turn, is the minimal left ideal of the biquaternion algebra $\C_2$. On the other hand, the algebra $\C_2$ is isomorphic to the complexifications of the real subalgebras $\cl_{2,0}$, $\cl_{1,1}$, $\cl_{0,2}$. Accordingly, the spinspace $\dS_2(\C)$ is isomorphic to the complexifications of minimal left ideals of real subalgebras\footnote{The minimal left ideal of the subalgebra $\cl_{2,0}$ has the form $I_{2,0}=\cl_{2,0}f_{20}$, where $f_{20}=\frac{1}{2}(1+\mathbf{e}_1)$ is the primitive idempotent for $\cl_{2,0}$, the corresponding division ring $\K\simeq f_{20}\cl_{2,0}f_{20}\simeq\{1\}\simeq\R$ is isomorphic to the field of real numbers. For the subalgebra $\cl_{1,1}$ we have $I_{1,1}=\cl_{1,1}f_{11}$, where $f_{11}=\frac{1}{2}(1+\mathbf{e}_{12})$, $\K\simeq f_{11}\cl_{1,1}f_{11}\simeq\{1\}\simeq\R$. Similarly, the minimal left ideal of the subalgebra $\cl_{0,2}$ is given by the expression $I_{0,2}=\cl_{0,2}f_{02}$, where $f_{02}=1$, and the quaternion division ring has the form $\K\simeq f_{02}\cl_{0,2}f_{02}\simeq\{1,\mathbf{e}_1,\mathbf{e}_2,\mathbf{e}_{12}\}\simeq\{1,\mathbf{i},\mathbf{j},\mathbf{k}\}\simeq\BH$. It should be noted that the choice of primitive idempotents of algebras $\cl_{p,q}$ is ambiguous and is determined by the number of commuting elements of the algebra $\cl_{p,q}$ forming a finite group of order $2^k$, where $k=q-r_{q-p}$, here $r_{q-p}$ are Radon-Hurwitz numbers whose values form a cycle with a period of 8: $r_{i+8}=r_i+4$ (for more details, see \cite{Lou,Var04}).}: $\dS_2(\C)\simeq\C\otimes I_{2,0}$, $\dS_2(\C)\simeq\C\otimes I_{1,1}$, $\dS_2(\C)\simeq\C\otimes I_{0,2}$. In addition, for the algebra $\C_2$, the isomorphism $\C_2\simeq\cl_{3,0}$ holds, where $\cl_{3,0}$ is a Clifford algebra over the field $\F=\R$ with a complex division ring $\K\simeq\C$ (type $p-q\equiv 3\pmod{8}$). In turn, for the algebra $\cl_{3,0}$, the decomposition $\cl_{3,0}\simeq\cl_{0,2}\oplus\cl_{0,2}\simeq\BH\oplus\BH$ is valid (a special case of the general isomorphism $\cl_{p,q}\simeq\cl_{q,p-1}\oplus\cl_{q,p-1}$ for odd-dimensional algebras, see \cite{Port,Var15} for more details). Further, since the maximal basis element $\omega=\mathbf{e}_1\mathbf{e}_2\mathbf{e}_3$ of the algebra $\cl_{3,0}$ belongs to the center $\mathbf{Z}_{3,0}=\{1,\omega\}\in\cl_{3,0}$ and $\omega^2=-1$, then $\mathbf{Z}_{3,0}\simeq\C$, and hence the algebra $\cl_{3,0}$ is a complexification (doubling) of the quaternion algebra $\cl_{0,2}\simeq\BH$, i.e., the biquaternion algebra $\C_2$. It follows that the complex structure (\ref{State}), which defines the charged states (according to Theorem 4), is a derivative of the neutral structure. In other words, the charge is a doubling (complexification) of the neutral structure\footnote{In this context, an electron is a complexification (doubling) of two neutrino states (neutrino and antineutrino), which is expressed in the standard formalism in the form of a Dirac bispinor. As is known \cite{BS}, the massless Dirac equation $i\gamma^\nu\partial_\nu\psi(x)=0$ in the ``split'' basis for the $\gamma$-matrices (the Weyl basis) decays into two equations $\left(\frac{\partial}{\partial x^0}\pm\boldsymbol{\sigma}\frac{\partial}{\partial\mathbf{x}}\right)\varphi_{(\pm)}(x)=0$, first proposed by Weyl in 1929 \cite{Weyl}, where the two-component functions $\varphi_\alpha$ [$\alpha=(-),(+)$] describe, respectively, the \textit{left-handed neutrino} (negative helicity, $\alpha=(-)$) and the \textit{right-handed antineutrino} (positive helicity, $\alpha=(+)$). In this case, the projection operators $P_\pm=1/2(1\pm\gamma_5)$ are the central idempotents of the Dirac algebra $\C_4$. Baez calls the absence of right-handed neutrinos and left-handed antineutrinos in nature the ``greatest mystery'' of the world \cite{Baez2}. The Weyl equations were obtained under the assumption that neutrinos have zero mass, however, according to modern data \cite{PDG}, neutrinos have non-zero mass, which depends on the variety (``flavour'') of neutrinos in the range: $m_{\nu_e}< 1,1$ eV, $m_{\nu_\mu}< 0,19$ MeV, $m_{\nu_\tau}< 18,2$ MeV. The presence of neutrino mass necessarily leads to the definition of Dirac-like equations for neutrinos, whose transformational properties (discrete symmetries) are significantly different from those for the Weyl equations. This expansion of transformational properties is in sharp contrast to the observational data. The existence of this contradiction is obviously a consequence of the attempt to describe quantum micro-objects within the framework of the continuum theory. The algebraic formulation, which is free from binding to the continuum, allows us to circumvent this contradiction.}. Thus, along with the fusion operation, we have the operation of doubling the states\footnote{According to de Broglie's neutrino theory of light \cite{Bro}, a photon is the result of the fusion of two neutrino states (neutrino and antineutrino): $\left|\gamma\right\rangle=\left|\nu\right\rangle\otimes\left|\bar{\nu}\right\rangle=\left|\BH\otimes\overline{\BH},0,0,1\right\rangle=
\left|\R,0,0,1\right\rangle$. In turn, the electron is a doubling (complexification) of the neutrino states: $\left|e\right\rangle=\left|\nu\right\rangle\oplus\left|\bar{\nu}\right\rangle=\left|\BH\oplus\overline{\BH},0,1,1/2\right\rangle=
\left|\C,0,1,1/2\right\rangle$. In this context, the electron and photon are \textit{derived states} (structures), and the \textit{truly fundamental state} is the neutrino. It is not difficult to see that all the cyclic vectors $\left|\psi\right\rangle$ of a $\K$-Hilbert space $\bsH_{\rm phys}(\K)$ can be obtained by fusion and doubling operations from the fundamental state $\left|\nu\right\rangle$. These operations define the dynamic relationship between the tensor structure and the $\K$-linear structure of the vectors $\left|\psi\right\rangle\in\bsH_{\rm phys}(\K)$, i.e., the relationship between the mass, spin, and charge of the states. In a sense, this makes it possible to generalize de Broglie's neutrino theory of light to the ``neutrino theory of everything''.}.

The fusion operation that generates state (\ref{State}) induces the tensor product $\K\otimes\K\otimes\cdots\otimes\K$ of the corresponding rings of fundamental states. According to the factorization (\ref{Factor}), any even-dimensional algebra $\cl_{p,q}$ over the field $\F=\R$ is isomorphic to the tensor product of two-dimensional algebras $\cl_{0,2}$ ($\K\simeq\BH$) and $\cl_{1,1}$, $\cl_{2,0}$ ($\K\simeq\R$). Accordingly, we have \textit{\textbf{two types of fundamental states}}\footnote{The second type corresponds to the algebras $\cl_{1,1}$, $\cl_{2,0}$ with the real ring $\K\simeq\R$. In this case, the charge conjugation $C$ (pseudo-automorphism $\cA\rightarrow\overline{\cA}$) is reduced to an identical transformation, i.e., the state coincides with its charge-conjugate state (anti-state). Thus, the second type sets the \textit{truly neutral state}, namely, the Majorana fermion of minimal spin $s=1/2$, i.e., the \textit{sterile neutrino} $\left|\nu_s\right\rangle$ (Majorana neutrino). The existence of a sterile (inert) neutrino was predicted by B.M. Pontecorvo \cite{Pont58}. Currently, an intensive search is underway for this type of neutrino, suggesting that $\left|\nu_s\right\rangle$ is the main component of the so-called ``dark matter''. In contrast to the inert neutrino $\left|\nu_s\right\rangle$, the \textit{active neutrino} $\left|\nu_a\right\rangle$ (type \textbf{I}) is divided into two types: the left-handed neutrino $\left|\nu_a\right\rangle$ and the right-handed antineutrino $\left|\bar{\nu}_a\right\rangle$.}:
\begin{description}
\item[I.] \phantom{I}$\left|\mathfrak{q}_a\right\rangle=\left|\BH,0,1,\frac{1}{2}\right\rangle$,
$\left|\bar{\mathfrak{q}}_a\right\rangle=\left|\overline{\BH},0,-1,\frac{1}{2}\right\rangle$.
\item[II.] $\left|\mathfrak{q}_s\right\rangle=\left|\R,0,0,\frac{1}{2}\right\rangle$,
$\left|\bar{\mathfrak{q}}_s\right\rangle=\left|\mathfrak{q}_s\right\rangle$.
\end{description}

Thus, according to (\ref{State}) and (\ref{Factor}), any state from $\bsH_{\rm phys}(\K)$ is a combination (fusion or doubling) active and inert fundamental states. Depending on $m\equiv 0,1\pmod{2}$ from (\ref{State}), we have fermionic or bosonic states. The fundamental states (types \textbf{I} and \textbf{II}) correspond to $m=1$. For $m=2$, the structure $\K\otimes\K$ corresponds to the algebras
\begin{eqnarray}
\cl_{4,0}&\simeq&\cl_{2,0}\otimes\cl_{0,2}\Rightarrow\R\otimes\BH\longrightarrow\BH,\nonumber\\
\cl_{3,1}&\simeq&\cl_{2,0}\otimes\cl_{1,1}\Rightarrow\R\otimes\R\longrightarrow\R,\nonumber\\
\cl_{2,2}&\simeq&\cl_{2,0}\otimes\cl_{2,0}\Rightarrow\R\otimes\R\longrightarrow\R,\nonumber\\
\cl_{1,3}&\simeq&\cl_{1,1}\otimes\cl_{0,2}\Rightarrow\R\otimes\BH\longrightarrow\BH,\nonumber\\
\cl_{0,4}&\simeq&\cl_{0,2}\otimes\cl_{2,0}\Rightarrow\BH\otimes\R\longrightarrow\BH.\nonumber
\end{eqnarray}
Further, for the structure $\K\otimes\K\otimes\K$ at $m=3$, we have\footnote{For the algebra $\cl_{4,2}\simeq\R(8)$, the decomposition $\cl_{4,2}\simeq\cl_{0,2}\otimes\cl_{0,4}\simeq\cl_{0,2}\otimes\cl_{0,2}\otimes\cl_{2,0}\simeq\BH\otimes
\BH\otimes\R(2)$ (Theorem 6) is valid, since $\BH\otimes\BH\simeq\R(4)$, then $\BH\otimes\BH\otimes\R\rightarrow\R$. Similarly for the algebras $\cl_{3,3}\simeq\R(8)$ and $\cl_{0,6}\simeq\R(8)$.}
\begin{eqnarray}
\cl_{6,0}&\simeq&\cl_{2,0}\otimes\cl_{0,2}\otimes\cl_{2,0}\Rightarrow
\R\otimes\BH\otimes\R\longrightarrow\BH,\nonumber\\
\cl_{5,1}&\simeq&\cl_{2,0}\otimes\cl_{1,1}\otimes\cl_{0,2}\Rightarrow
\R\otimes\R\otimes\BH\longrightarrow\BH,\nonumber
\end{eqnarray}
\begin{eqnarray}
\cl_{4,2}&\simeq&\cl_{2,0}\otimes\cl_{2,0}\otimes\cl_{2,0}\Rightarrow
\R\otimes\R\otimes\R\longrightarrow\R,\nonumber\\
&\simeq&\cl_{1,1}\otimes\cl_{2,0}\otimes\cl_{1,1}\Rightarrow
\R\otimes\R\otimes\R\longrightarrow\R,\nonumber\\
&\simeq&\cl_{0,2}\otimes\cl_{0,2}\otimes\cl_{2,0}\Rightarrow
\BH\otimes\BH\otimes\R\longrightarrow\R,\nonumber
\end{eqnarray}
\begin{eqnarray}
\cl_{3,3}&\simeq&\cl_{2,0}\otimes\cl_{2,0}\otimes\cl_{1,1}\Rightarrow
\R\otimes\R\otimes\R\longrightarrow\R,\nonumber\\
&\simeq&\cl_{0,2}\otimes\cl_{1,1}\otimes\cl_{0,2}\Rightarrow
\BH\otimes\R\otimes\BH\longrightarrow\R,\nonumber
\end{eqnarray}
\begin{eqnarray}
\cl_{2,4}&\simeq&\cl_{2,0}\otimes\cl_{2,0}\otimes\cl_{0,2}\Rightarrow
\R\otimes\R\otimes\BH\longrightarrow\BH,\nonumber\\
&\simeq&\cl_{1,1}\otimes\cl_{1,1}\otimes\cl_{0,2}\Rightarrow
\R\otimes\R\otimes\BH\longrightarrow\BH,\nonumber
\end{eqnarray}
\begin{eqnarray}
\cl_{1,5}&\simeq&\cl_{1,1}\otimes\cl_{0,2}\otimes\cl_{2,0}\Rightarrow
\R\otimes\BH\otimes\R\longrightarrow\BH,\nonumber\\
\cl_{0,6}&\simeq&\cl_{0,2}\otimes\cl_{2,0}\otimes\cl_{0,2}\Rightarrow
\BH\otimes\R\otimes\BH\longrightarrow\R.\nonumber
\end{eqnarray}
At $m=4$ for $\K\otimes\K\otimes\K\otimes\K$ we get
\begin{eqnarray}
\cl_{8,0}&\simeq&\cl_{2,0}\otimes\cl_{0,2}\otimes\cl_{2,0}\otimes\cl_{0,2}
\Rightarrow\R\otimes\BH\otimes\R\otimes\BH\longrightarrow\R,\nonumber\\
\cl_{7,1}&\simeq&\cl_{2,0}\otimes\cl_{1,1}\otimes\cl_{0,2}\otimes\cl_{2,0}
\Rightarrow\R\otimes\R\otimes\BH\otimes\R\longrightarrow\BH,\nonumber
\end{eqnarray}
\begin{eqnarray}
\cl_{6,2}&\simeq&\cl_{2,0}\otimes\cl_{2,0}\otimes\cl_{2,0}\otimes\cl_{0,2}
\Rightarrow\R\otimes\R\otimes\R\otimes\BH\longrightarrow\BH,\nonumber\\
&\simeq&\cl_{1,1}\otimes\cl_{2,0}\otimes\cl_{1,1}\otimes\cl_{0,2}
\Rightarrow\R\otimes\R\otimes\R\otimes\BH\longrightarrow\BH,\nonumber\\
&\simeq&\cl_{0,2}\otimes\cl_{0,2}\otimes\cl_{2,0}\otimes\cl_{0,2}
\Rightarrow\BH\otimes\BH\otimes\R\otimes\BH\longrightarrow\BH,\nonumber
\end{eqnarray}
\begin{eqnarray}
\cl_{5,3}&\simeq&\cl_{2,0}\otimes\cl_{2,0}\otimes\cl_{2,0}\otimes\cl_{1,1}
\Rightarrow\R\otimes\R\otimes\R\otimes\R\longrightarrow\R,\nonumber\\
\cl_{4,4}&\simeq&\cl_{2,0}\otimes\cl_{2,0}\otimes\cl_{2,0}\otimes\cl_{2,0}
\Rightarrow\R\otimes\R\otimes\R\otimes\R\longrightarrow\R,\nonumber\\
&\simeq&\cl_{1,1}\otimes\cl_{2,0}\otimes\cl_{2,0}\otimes\cl_{1,1}
\Rightarrow\R\otimes\R\otimes\R\otimes\R\longrightarrow\R,\nonumber\\
&\simeq&\cl_{0,2}\otimes\cl_{2,0}\otimes\cl_{2,0}\otimes\cl_{0,2}
\Rightarrow\BH\otimes\R\otimes\R\otimes\BH\longrightarrow\R,\nonumber
\end{eqnarray}
\begin{eqnarray}
\cl_{3,5}&\simeq&\cl_{2,0}\otimes\cl_{2,0}\otimes\cl_{1,1}\otimes\cl_{0,2}
\Rightarrow\R\otimes\R\otimes\R\otimes\BH\longrightarrow\BH,\nonumber\\
\cl_{2,6}&\simeq&\cl_{2,0}\otimes\cl_{2,0}\otimes\cl_{0,2}\otimes\cl_{2,0}
\Rightarrow\R\otimes\R\otimes\BH\otimes\R\longrightarrow\BH,\nonumber\\
&\simeq&\cl_{1,1}\otimes\cl_{1,1}\otimes\cl_{0,2}\otimes\cl_{2,0}
\Rightarrow\R\otimes\R\otimes\BH\otimes\R\longrightarrow\BH,\nonumber\\
\cl_{1,7}&\simeq&\cl_{1,1}\otimes\cl_{0,2}\otimes\cl_{2,0}\otimes\cl_{0,2}
\Rightarrow\R\otimes\BH\otimes\R\otimes\BH\longrightarrow\R,\nonumber\\
\cl_{0,8}&\simeq&\cl_{0,2}\otimes\cl_{2,0}\otimes\cl_{0,2}\otimes\cl_{2,0}
\Rightarrow\BH\otimes\R\otimes\BH\otimes\R\longrightarrow\R.\nonumber
\end{eqnarray}
For $m=5,6,\ldots$, the explicit form of the structure $\K\otimes\K\otimes\cdots\otimes\K$ is determined by the Cartan-Bott periodicity\footnote{In 1908, Cartan \cite{Car08} identified the Clifford algebras as matrix algebras with entries in $\R$, $\C$, $\BH$, $\R\oplus\R$, $\BH\oplus\BH$ and found a periodicity of 8. Cartan's periodicity of 8 for Clifford algebras is often attributed to Bott, who proved his periodicity of homotopy groups of rotation groups in 1959 \cite{Bot59}.} for algebras $\cl_{p,q}$ over the field $\F=\R$ (for more details, see \cite{Var15}). The Clifford algebra over the field $\F=\R$ is modulo 8 periodic: $\cl_{p+8,q}\simeq\cl_{p,q}\otimes\cl_{8,0}$ ($\cl_{p,q+8}\simeq\cl_{p,q}\otimes\cl_{0,8}$). So, at $m=5$ for the algebra $\cl_{10,0}\simeq\R(32)$ we have the decomposition
\[
\cl_{10,0}\simeq\cl_{2,0}\otimes\cl_{8,0}\simeq\R(2)\otimes\R(16)\Rightarrow
\R\otimes\BH\otimes\R\otimes\BH\otimes\R\longrightarrow\R
\]
that corresponds to
\[
\cl_{2,0}\otimes\cl_{0,2}\otimes\cl_{2,0}\otimes\cl_{0,2}\otimes\cl_{2,0}\longrightarrow
\left|\mathfrak{q}_s\right\rangle\otimes\left|\mathfrak{q}_a\right\rangle\otimes\left|\mathfrak{q}_s\right\rangle
\otimes\left|\mathfrak{q}_a\right\rangle\otimes\left|\mathfrak{q}_s\right\rangle
\]
and so on. Thus, all the states that make up the $\BH$- and $\R$-subspaces of the physical $\K$-Hilbert space can be constructed by means of the fusion operation from the fundamental states.

The states that make up the $\C$-subspaces are determined by the doubling operation. In this case, the key role is played by the isomorphism $\C_{p+q-1}\simeq\cl_{p,q}$, where $p-q\equiv 3,7\pmod{8}$. A special case is the isomorphism $\C_2\simeq\cl_{3,0}$ discussed above. In turn, for $p-q\equiv 3,7\pmod{8}$, for the odd-dimensional algebra $\cl_{p,q}$ the decomposition\footnote{For the algebras of the form $\cl_{0,q}$, we have $\cl_{0,q}\simeq\cl_{0,q-1}\oplus\cl_{0,q-1}$.}
\begin{equation}\label{Decomp}
\cl_{p,q}\simeq\cl_{q,p-1}\oplus\cl_{q,p-1},
\end{equation}
is valid, which can be represented by the following scheme:
\[
\unitlength=0.5mm
\begin{picture}(70,50)
\put(35,40){\vector(2,-3){15}} \put(35,40){\vector(-2,-3){15}}
\put(28.25,42){$\cl_{p,q}$} \put(16,28){$\lambda_{+}$}
\put(49.5,28){$\lambda_{-}$} \put(9.5,9.20){$\cl_{q,p-1}$}
\put(47.75,9){$\cl_{q,p-1}$} \put(32.5,10){$\oplus$}
\end{picture}
\]
Here, the central idempotents\footnote{According to \cite{CF96}, the idempotents $\lambda_+$ and $\lambda_-$ can be identified with the projection operators of helicity, which distinguish between left- and right-polarized spinors.}
\[
\lambda_+=\frac{1+\e_1\e_2\cdots\e_{p+q}}{2},\quad\lambda_-=\frac{1-\e_1\e_2\cdots\e_{p+q}}{2}
\]
satisfy the realtions $(\lambda_+)^2=\lambda_+$, $(\lambda_-)^2=\lambda_-$, $\lambda_+\lambda_-=0$. For $p-q\equiv 3\pmod{8}$, the isomorphism (\ref{Decomp}) leads to the decomposition of the algebra $\cl_{p,q}$ into two subalgebras $\cl_{q,p-1}$ with a quaternion division ring $\K\simeq\BH$:
\[
\cl_{p,q}\simeq\cl_{q,p-1}\oplus\cl_{q,p-1}\simeq\BH\left(2^{\frac{p+q-3}{2}}\right)\oplus
\BH\left(2^{\frac{p+q-3}{2}}\right)\longrightarrow\BH\oplus\BH.
\]
Further, for $p-q\equiv 7\pmod{8}$, the isomorphism (\ref{Decomp}) leads to the decomposition of the algebra $\cl_{p,q}$ into two subalgebras $\cl_{q,p-1}$ with a real division ring $\K\simeq\R$:
\[
\cl_{p,q}\simeq\cl_{q,p-1}\oplus\cl_{q,p-1}\simeq\R\left(2^{\frac{p+q-1}{2}}\right)\oplus
\R\left(2^{\frac{p+q-1}{2}}\right)\longrightarrow\R\oplus\R.
\]
In both cases, the maximal basis element $\omega=\e_1\e_2\cdots\e_{p+q}$ of the algebras $\cl_{p,q}$ belongs to the center $\mathbf{Z}_{p,q}=\{1,\omega\}$. Since $\omega^2=-1$, then $\mathbf{Z}_{p,q}\simeq\C$, and hence the algebras $\cl_{p,q}$ are complexifications (doublings) of their subalgebras $\cl_{q,p-1}$ with rings $\K\simeq\BH$ and $\K\simeq\R$: $\BH\oplus i\BH$ ($p-q\equiv 3\pmod{8}$) and $\R\oplus i\R$ ($p-q\equiv 7\pmod{8}$). Thus, we have \textit{two types of charge}: 1) $\C\simeq\BH\oplus i\BH$ -- \textit{doubling of active states}; 2) $\C\simeq\R\oplus i\R$ -- \textit{doubling of inert states}.
\end{proof}

\section{Summary}
The main premise of the algebraic formulation of quantum theory is the possibility of constructing a theory without involving any classical analogies and related visual images and mechanical models. The construction of a quantum theory should be carried out exclusively by means of its mathematical apparatus. Any macroscopic analogies introduced from classical physics should be discarded. As for subatomic physics, the construction of the theory should mainly rely on the group-theoretical (symmetric) method.

According to the algebraic interpretation, the set of observables forms an algebra $\fA$, in which the operation of multiplying the observables is defined and their linear superpositions are given. The explicit relationship between the algebra $\fA$ and the measurement data is given by the concept of the state $\omega$, by which the expected value $\omega(\fa)$ of the observable $\fa\in\fA$ can be determined. The canonical correspondence $\omega\leftrightarrow\pi_\omega$ between states and cyclic representations of the $C^\ast$-algebra $\fA$ is given by the GNS (Gelfand-Naimark-Segal) construction. Following Heisenberg \cite{Heisen51}, we believe that at the fundamental level, the main observable is the energy that corresponds to the Hermitian operator $H$. The group $\SU(2,2)$ is chosen as the fundamental symmetry that allows structuring the energy levels of the state spectrum. Thus, the $C^\ast$-algebra $\fA$ consists of the energy operator $H$ and the generators of the group $\SU(2,2)$ attached to $H$, forming a common system of eigenfunctions with $H$. The set $\Omega$ of pure separable states $\omega$ on the algebra $\fA$ corresponds to a system of cyclic vectors of the form
\begin{equation}\label{Cycle2}
\pi_\omega(\fh^{(1)})\pi_\omega(\fh^{(2)})\cdots\pi_\omega(\fh^{(n)}
\left|\omega\right\rangle\;\longmapsto\;\boldsymbol{\tau}_{\frac{k}{2},0}\otimes
\boldsymbol{\tau}_{0,\frac{r}{2}}\left|\omega\right\rangle
\end{equation}
in the $\K$-Hilbert space $\bsH(\K)$ (see Theorems 3 and 4), where $\fh^{(i)}\in\fA$, $i=1,2,\ldots,n$; $\pi_\omega\simeq\boldsymbol{\tau}_{l\dot{l}}$ is a representation of the spinor group $\spin_+(1,3)$, $\K=\R,\C,\BH$. The physical $\K$-Hilbert space $\bsH_{\rm phys}(\K)$ (spectrum of matter (energy)) is formed by cyclic vectors (\ref{Cycle2}), for which the mass $m_\omega$ of the pure separable state $\omega$ is determined by the formula
$
m_{\omega}=m_e\left(l+1/2\right)(\dot{l}+1/2),
$
here $m_e$ is the rest mass of the electron (for more details, see \cite{Var17}). Thus, the mass (energy) of states $\omega$ is given by the tensor structure of cyclic vectors $\left|\psi\right\rangle\in\bsH_{\rm phys}(\K)$. The charge of the state $\omega$ is defined within the $\K$-linear structure of the space $\bsH_{\rm phys}(\K)$: $\K\simeq\C$ -- charged states, $\K\simeq\BH$ -- neutral states, $\K\simeq\R$ -- truly neutral states. In contrast to the mechanical definition of the spin of a quantum micro-object ($q\bar{q}$-meson or $qqq$-baryon) of the quark model, in the algebraic formulation we have a non-classical (group-theoretic) definition of this most important characteristic: $s=l-\dot{l}$. The entire spectrum of states is divided into a sequence of spin lines, along which the states have the same spin, but different mass (tensor structure). Other important characteristics of the state (along with mass, charge and spin) are discrete symmetries and associated quantum numbers ($P$-parity, $C$-parity, etc.). And here, in contrast to the definitions that depend on the angular moments of the mechanical model of quarks, the algebraic approach is more universal and does not depend on any classical (macroscopic) definitions. Namely, the Clifford algebra $\cl$ is associated with each cyclic vector (\ref{Cycle}). In the case of a number field $\F=\C$, eight automorphisms are defined for the Clifford algebra $\C_n$  \cite{Rash,Var03} (including the identical automorphism $\Id$)\footnote{In 1955, Rashevsky \cite{Rash} showed that there are \textit{four fundamental automorphisms} of the algebra $\C_n$: $\cA\rightarrow\cA$ (identity), $\cA\rightarrow\cA^\star$ (involution), $\cA\rightarrow\widetilde{\cA}$ (reversion) and $\cA\rightarrow\widetilde{\cA^\star}$ (conjugation), where $\cA$ is an arbitrary element of the algebra $\C_n$. The group structure of the set of automorphisms $\{\Id,\,\star,\,\widetilde{\phantom{cc}},\,\widetilde{\star}\}$ with respect to discrete transformations that make up the $PT$ group (the so-called \textit{reflection group}) was studied in \cite{Var99,Var01}. Along with the fundamental automorphisms of the algebra $\C_n$ there is a \textit{pseudo-automorphism} $\cA\rightarrow\overline{\cA}$ (for more details, see \cite{Rash}), which is not fundamental, but its composition with the fundamental automorphisms allows us to extend the set $\{\Id,\,\star,\,\widetilde{\phantom{cc}},\,\widetilde{\star}\}$ by means of pseudo-automorphisms $\cA\rightarrow\overline{\cA}$, $\cA\rightarrow\overline{\cA^\star}$, $\cA\rightarrow\overline{\widetilde{\cA}}$, $\cA\rightarrow\overline{\widetilde{\cA^\star}}$. The group structure of the \textit{extended set of automorphisms}, $\{\Id,\,\star,\,\widetilde{\phantom{cc}},\,\widetilde{\star},\,
\overline{\phantom{cc}},\,\overline{\star},\,
\overline{\widetilde{\phantom{cc}}},\,\overline{\widetilde{\star}}\}$, was studied in \cite{Var03,Var04} in relation to $CPT$ symmetries.}. There is an isomorphism between $\Ext(\C_n)=\{\Id,\,\star,\,\widetilde{\phantom{cc}},\,\widetilde{\star},\,
\overline{\phantom{cc}},\,\overline{\star},\,\overline{\widetilde{\phantom{cc}}},\,\overline{\widetilde{\star}}\}$ and the $CPT$ group of discrete transformations\footnote{It is interesting to note that in the Penrose twistor program \cite{Pen77}, the spinor structure is understood as the underlying (more fundamental, primary) structure with respect to Minkowski space-time. In other words, the space-time continuum is not a fundamental substance in the twistor approach. The continuum is an absolutely derived entity (in the spirit of Leibnitz's relational philosophy) generated by the underlying spinor structure. In this context, the space-time discrete symmetries $P$ and $T$ are projections (shadows) of the fundamental automorphisms of the spinor structure.}. At this point, the inversion of space $P$, the time reversal $T$, the complete reflection $PT$, the charge conjugation $C$, the transformations $CP$, $CT$ and the complete $CPT$ transformation correspond to the automorphism $\cA\rightarrow\cA^\star$, the anti-automorphisms $\cA\rightarrow\widetilde{\cA}$, $\cA\rightarrow\widetilde{\cA^\star}$, the pseudo-automorphisms $\cA\rightarrow\overline{\cA}$, $\cA\rightarrow\overline{\cA^\star}$, the pseudo-anti-automorphisms
$\cA\rightarrow\overline{\widetilde{\cA}}$ and $\cA\rightarrow\overline{\widetilde{\cA^\star}}$, respectively
\cite{Var03,Var04}.

In conclusion it should be noted that the main purpose of this article was the desire to emphasize the fundamental role of two-component spinors. A two-component spinor describes the minimal structural component of matter (a quantum of energy). Identification of the minimal structural component $\left|\mathfrak{q}\right\rangle$ (quantum of energy) with neutrino $\left|\nu\right\rangle$, $\left|\mathfrak{q}\right\rangle\equiv\left|\nu\right\rangle$, would lead to the ``neutrino theory of everything'' and to the understanding of neutrino as the primary element of matter. However, the substance (energy) is not determined by its states. As Spinoza said: ``Substance is by nature the first of its states'' (Theorem 1, Ethics). Substance as a whole is more primary than its states. States have a secondary (subordinate) nature in relation to the whole (substance). This is one of the main principles of holism. The spectrum of states of the substance, which we call energy or matter, contains an infinite number of states, among which one of the most minimal is the neutrino. Substance is not an aggregate constructed from elementary parts, contrary to all the ideas of atomism and reductionism. The quantized nature of the Fermi and Bose states of the matter spectrum consists in the factorization (separability) of cyclic vectors of the $\K$-Hilbert space by means of the tensor product of fundamental states $\left|\mathfrak{q}\right\rangle$ (GNS construction, algebraic quantization). At the same time, the presence of a $\K$-structure splits $\left|\mathfrak{q}\right\rangle$ into two states $\left|\mathfrak{q}_a\right\rangle$ and $\left|\mathfrak{q}_s\right\rangle$.

\end{document}